\documentclass[12pt]{article}

\usepackage{ a4, graphicx, amsmath, amsthm, color, ulem, psfrag, url }

\usepackage{ amssymb, bigints }

\newcommand{ \dn }[ 1 ]{ \boldsymbol #1 }
\newcommand{ \mt }[ 1 ]{ \mathrm #1 }

\newcommand{ \be }{ \begin{equation} }
\newcommand{ \ee }{ \end{equation} }
\newcommand{ \bmath }{ \begin{displaymath} }
\newcommand{ \emath }{ \end{displaymath} }

\newcommand{ \by }{ \dn{ y } }

\newcommand{ \ident }{ \mt{ I }_{ n } }
\newcommand{ \identst }{ \mt{ I }_{ \nstar } }
\newcommand{ \nstar }{ n^* }

\newcommand{ \bbl }{ \dn{ \beta }_{ \ell } }
\newcommand{ \bbo }{ \dn{ \beta }_{ 0 } }

\newcommand{ \bx }{ \mt{ X } }

\newcommand{ \bxl }{ \bx_{\ell} }
\newcommand{ \bxo }{ {\bx}_{0} }

\newcommand{ \bboup }{ \overline{\overline{ \dn{ \beta }} }_{ 0 } }

\newcommand{ \varl }{ \sigma_\ell^2 }
\newcommand{ \varo }{ \sigma_0^2 }

\newcommand{ \xtxstarl }{ \big( {\mt{X}_\ell^{*}}^T \mt{X}_\ell^* \big) }

\newcommand{ \invxtxstarl }{ \xtxstarl^{-1} }

\newtheoremstyle{mytheoremstyle}{3pt}{3pt}{\itshape}{}{\scshape}{:}{0.5em}{}
\theoremstyle{mytheoremstyle}
\newtheorem{theorem}{Theorem}
\newtheorem{lemma}{Lemma}

\usepackage[comma,longnamesfirst]{natbib}

\renewcommand{\citename}{\citeauthor*}

\begin{document}

\normalem

\title{\vspace*{-0.7in} \textbf{Power-Expected-Posterior Priors \\ for
Variable Selection \\ in Gaussian Linear Models}}

\author{
D.~Fouskakis\thanks{D.~Fouskakis is with the Department of Mathematics,
National Technical University of Athens, Zografou Campus, Athens 15780
Greece; email \texttt{fouskakis@math.ntua.gr}}, \
I.~Ntzoufras\thanks{I.~Ntzoufras is with the Department of Statistics,
Athens University of Economics and Business, 76 Patision Street, Athens
10434 Greece; email \texttt{ntzoufras@aueb.gr}} \ and
D.~Draper\thanks{D.~Draper is with the Department of Applied Mathematics
and Statistics, Baskin School of Engineering, University of California,
1156 High Street, Santa Cruz CA 95064 USA; email
\texttt{draper@ams.ucsc.edu} and with the \textit{Center of
Excellence in Statistical Research}, \textsc{ebay research labs}, San Jose
CA USA}
}

\date{\small May 15, 2014}

\maketitle

\vspace*{-0.15in}

\noindent
\textbf{Summary:} In the context of the expected-posterior prior (EPP)
approach to Bayesian variable selection in linear models, we combine ideas
from power-prior and unit-information-prior methodologies to simultaneously
(a) produce a minimally-informative prior and (b) diminish the effect of
training samples. The result is that in practice our
\textit{power-expected-posterior} (PEP) methodology is sufficiently
insensitive to the size $n^*$ of the training sample, due to PEP's
unit-information construction, that one may take $n^*$ equal to the
full-data sample size $n$ and dispense with training samples altogether. 
This promotes stability of the resulting Bayes factors, removes the
arbitrariness arising from individual training-sample selections, and
greatly increases computational speed, allowing many more models to be
compared within a fixed CPU budget. In this paper we focus on Gaussian
linear models and develop our PEP methodology under two different baseline
prior choices: the independence Jeffreys (or reference) prior, yielding the
\textit{J-PEP} posterior, and the Zellner $g$-prior, leading to
\textit{Z-PEP}. We find that, under the reference baseline prior, the
asymptotics of PEP Bayes factors are equivalent to those of Schwartz's BIC
criterion, ensuring consistency of the PEP approach to model selection. We
compare the performance of our method, in simulation studies and a real
example involving prediction of air-pollutant concentrations from
meteorological covariates, with that of a variety of previously-defined
variants on Bayes factors (and other methods) for objective variable
selection. Our PEP prior, due to its unit-information structure, leads to a
variable-selection procedure that --- in our empirical studies --- (1) is
systematically more parsimonious than the basic EPP with minimal training
sample, while sacrificing no desirable performance characteristics to
achieve this parsimony; (2) is robust to the size of the training sample,
thus enjoying the advantages described above arising from the avoidance of
training samples altogether; and (3) identifies maximum-a-posteriori models
that achieve better out-of-sample predictive performance than that provided
by standard EPPs, the $g$-prior, the hyper-$g$ prior, non-local priors, the
LASSO and SCAD. Moreover, PEP priors are diffuse even when $n$ is not much
larger than the number of covariates $p$, a setting in which EPPs can be
far more informative than intended.

\vspace*{0.15in}

\noindent
\textit{Keywords:} Bayesian variable selection; Bayes factors; Consistency;
Expected-posterior priors; Gaussian linear models; $g$-prior; Hyper-$g$
prior; LASSO; Non-local priors; Prior compatibility; Power-prior; Training
samples; SCAD; Unit-information prior

\vspace*{0.15in}

\noindent
\textit{Note:} A glossary of abbreviations is given before the references
at the end of the paper.

\section{Introduction}

\label{introduction}

A leading approach to Bayesian variable selection in regression models is
based on posterior model probabilities and the corresponding posterior
model odds, which are functions of Bayes factors. In the case of Gaussian
regression models, on which we focus in this paper, an active area of
research has emerged from attempts to use improper prior distributions in
this approach; influential contributions include a variety of Bayes-factor
variants (\textit{posterior}, \textit{fractional} and \textit{intrinsic}:
see, e.g., \cite{aitkin_91}, \cite{ohagan_95}, and
\citename{berger_pericchi_96a}
(\citeyear{berger_pericchi_96a,berger_pericchi_96b}), respectively).

An important part of this work is focused on \textit{objective model
selection methods} (\cite{casella_moreno_2006}; \cite{moreno_giron_2008};
\cite{casella_etal_2009}), having their source in the intrinsic priors
originally introduced by \cite{berger_pericchi_96b}; these methods attempt
to provide an approximate proper Bayesian interpretation for intrinsic
Bayes factors (IBFs). Intrinsic priors can be considered as special cases
of the \textit{expected-posterior prior} (EPP) distributions of
\cite{perez_berger_2002}, which have an appealing interpretation based on
imaginary training data coming from prior predictive distributions. EPP
distributions can accommodate improper \textit{baseline} priors as a
starting point, and the marginal likelihoods for all models are calculated
up to the same normalizing constant; this overcomes the problem of
indeterminacy of the Bayes factors. Moreover, as
\cite{consonni_veronese_2008} note, ``EPP is a method to make priors
compatible across models, through their dependence on a common marginal
data distribution; thus this methodology can be applied also with
subjectively specified (proper) prior distributions." However, in
regression problems, the approach is based on one or more \textit{training
samples} chosen from the data, and this raises three new questions: how
large should such training samples be, how should they be chosen, and how
much do they influence the resulting posterior distributions?

In this paper we develop a minimally-informative prior and simultaneously
diminish the effect of training samples on the EPP approach, by combining
ideas from the power-prior method of \cite{ibrahim_chen_2000} and the
unit-information-prior approach of \cite{kass_wasserman_95}: we raise the
likelihood involved in the EPP distribution to the power $\frac{ 1 }{ n }$
(where $n$ denotes the sample size), to produce a prior information content
equivalent to one data point. In this manner the effect of the
imaginary/training sample is small with even modest $n$. Moreover, as will
become clear in Section \ref{examples}, in practice our
\textit{power-expected-posterior} (PEP) prior methodology, due to its
low-information structure, is sufficiently insensitive to the size $n^*$ of
the training sample that one may take $n^* = n$ and dispense with training
samples altogether; this both removes the instability arising from the
random choice of training samples and greatly reduces computation time.

As will be seen, PEP priors have an additional advantage over standard EPPs
in settings, which arise with some frequency in disciplines such as
bioinformatics/genomics (e.g., \cite{national2005Mathematics}) and
econometrics (e.g., \cite{johnstone-titterington-2009}), in which $n$ is not
much larger than the number of covariates $p$: standard EPPs can be far more
informative than intended in such situations, but the unit-information
character of PEP priors ensures that this problem does not arise with the
PEP approach.

PEP methodology can be implemented under any baseline prior choice, proper
or improper. In this paper, results are presented for two different prior
baseline choices: the Zellner $g$-prior and the independence Jeffreys
prior. The conjugacy structure of the first of these choices (a) greatly
increases calculation speed and (b) permits computation of the first two
moments (see Section 1 of the web Appendix) of the resulting PEP prior,
which offers flexibility in situations in which non-diffuse parametric
prior information is available. When (on the other hand) little
information, external to the present data set, about the parameters in the
competing models is available, the PEP prior with the independence Jeffreys
(or reference) baseline prior can be viewed as an \textit{objective
model-selection} technique, and the fact that the PEP posterior with the
Jeffreys baseline is a special case of the posterior with the $g$-prior as
baseline provides significant computational acceleration using the Jeffreys
baseline.

With either choice of baseline prior, simple but efficient Monte-Carlo
schemes for the estimation of the marginal likelihoods can be constructed
in a straightforward manner. We find that the corresponding Bayes factors,
under the reference baseline prior, are asymptotically equivalent to those
of the BIC criterion \citep{schwarz_78}; therefore the resulting PEP
objective Bayesian variable-selection procedure is consistent.

We wish to emphasize two points, at the outset, regarding our intentions in
developing PEP.

\begin{itemize}

\item

The purpose of the paper is not to compare the performance of PEP priors
with that of approaches such as mixtures of $g$-priors (e.g.,
\cite{liang_etal_2008}) or BIC itself. The point here is to begin with
EPPs, which are in wide use and which have the important property of
compatibility across models (a feature that mixtures of $g$-priors lack),
and to substantially improve EPPs by overcoming the difficulties that arise
from their dependence on training samples.

\item

The paper focuses on a variable-selection problem in the class of linear
models with fixed covariate space, where the number of available covariates
is less than the sample size ($p < n$); we do not intend this method to be
used in settings in which $p > n$.

\end{itemize}

The plan of the remainder of the paper is as follows. In the next two
sub-Sections, to fix notation and ideas, we provide some preliminary details
on the EPP approach, and we highlight difficulties that arise when
implementing it in variable-selection problems. Our PEP prior methodology is
described in detail in Section \ref{sec_pep}, and the resulting prior and
posterior distributions are presented under the two different baseline prior
choices mentioned above. In Section \ref{marginal_likelihood_pep_gen} we
provide Monte-Carlo estimates of the marginal likelihood for our approach. 
Section \ref{Limit} explores the limiting behavior of the resulting Bayes
factors, under the reference baseline prior. In Section \ref{examples} we
present illustrations of our method, under both baseline prior choices, in a
simulation experiment and in a real-data example involving the prediction of
atmospheric ozone levels from meteorological covariates; we also compare PEP
with seven other variable-selection and coefficient-shrinkage methods on
out-of-sample predictive performance. Finally, Section \ref{sec_discussion}
concludes the paper with a brief summary and some ideas for further
research. 

\subsection{Expected-posterior priors}

\label{expected-posterior-priors}

\cite{perez_berger_2002} developed priors for use in model comparison,
through utilization of the device of ``imaginary training samples"
(\cite{good}; \cite{spiegelhalter_smith_80}; \cite{iwaki}). They defined the
expected-posterior prior (EPP) as the posterior distribution of a parameter
vector for the model under consideration, averaged over all possible
imaginary samples $\dn{ y }^*$ coming from a ``suitable" predictive
distribution $m^* ( \dn{ y }^* )$. Hence the EPP for the parameter vector
$\dn{ \theta }_\ell$ of any model $M_\ell \in \cal{ M }$, with ${ \cal M }$
denoting the model space, is
\begin{equation} \label{epp} 
\pi^E_\ell ( \dn{ \theta }_\ell ) = \int \pi_\ell^N ( \dn{ \theta }_\ell |
\dn{ y }^* ) \, m^* ( \dn{ y }^* ) \, d \dn{ y }^* \, , 
\end{equation} 
where $\pi_\ell^N ( \dn{ \theta }_\ell | \dn{ y }^* )$ is the posterior
$\dn{ \theta }_\ell$ using a baseline prior $\pi_\ell^N ( \dn{ \theta }_\ell
)$ and data $\dn{ y }^*$.

A question that naturally arises when using EPPs is which predictive
distribution $m^*$ to employ for the imaginary data $\dn{ y }^*$ in
(\ref{epp}); \cite{perez_berger_2002} discussed several choices for $m^*$.
An attractive option, leading to the so-called \textit{base-model
approach}, arises from selecting a ``reference" or ``base" model $M_0$ for
the training sample and defining $m^* ( \dn{ y }^* ) = m_0^N ( \dn{ y }^* )
\equiv f ( \dn{ y }^* | M_0)$ to be the prior predictive distribution,
evaluated at $\dn{ y }^*$, for the reference model $M_0$ under the baseline
prior $\pi_0^N ( \dn{ \theta }_0 )$. Then, for the reference model (i.e.,
when $M_\ell = M_0$), (\ref{epp}) reduces to $\pi^E_0 ( \dn{ \theta }_0 ) =
\pi_0^N ( \dn{ \theta }_0 )$. Intuitively, the reference model should be at
least as simple as the other competing models, and therefore a reasonable
choice is to take $M_0$ to be a common sub-model of all $M_\ell \in { \cal
M }$. This interpretation is close to the skeptical-prior approach
described by \citet[Section 5.5.2]{spiegelhalter_abrams_myles}, in which a
tendency toward the null hypothesis can be a-priori supported by centering
the prior around values assumed by this hypothesis when no other
information is available. In the variable-selection problem that we
consider in this paper, the constant model (with no predictors) is clearly
a good reference model that is nested in all the models under
consideration. This selection makes calculations simpler, and additionally
makes the EPP approach essentially equivalent to the arithmetic intrinsic
Bayes factor approach of \cite{berger_pericchi_96a}.

One of the advantages of using EPPs is that impropriety of baseline priors
causes no indeterminacy. There is no problem with the use of an improper
baseline prior $\pi_\ell^N ( \dn{ \theta }_\ell )$ in (\ref{epp}); the
arbitrary constants cancel out in the calculation of any Bayes factor.
Impropriety in $m^*$ also does not cause indeterminacy, because $m^*$ is
common to the EPPs for all models. When a proper baseline prior is used, the
EPP and the corresponding Bayes factors will be relatively insensitive to
large values of the prior variances of the components of $\dn{ \theta
}_\ell$.

\subsection{EPPs for variable selection in Gaussian linear models}

\label{epp-gaussian-regression}

In what follows, we examine variable-selection problems in Gaussian
regression models. We consider two models $M_\ell$ (for $\ell = 0, 1$) with
parameters $\dn{ \theta }_\ell = ( \dn{ \beta }_\ell \, , \sigma_\ell^2 )$
and likelihood specified by
\begin{equation} \label{new2-1}
( \dn{ Y } | \mt{ X }_\ell, \dn{ \beta }_\ell, \sigma_\ell^2, M_\ell ) \sim
N_n ( \mt{ X }_\ell \, \dn{ \beta }_\ell \, , \sigma_\ell^2 \, \mt{ I }_n )
\, ,
\end{equation}
where $\dn{ Y } = ( Y_1, \dots, Y_n )$ is a vector containing the
(real-valued) responses for all subjects, $\mt{ X }_\ell$ is an $n \times
d_\ell$ design matrix containing the values of the explanatory variables in
its columns, $\mt{ I }_n$ is the $n \times n$ identity matrix, $\dn{ \beta
}_\ell$ is a vector of length $d_\ell$ summarizing the effects of the
covariates in model $M_\ell$ on the response $\dn{ Y }$ and $\sigma_\ell^2$
is the error variance for model $M_\ell$. Variable selection based on EPP
was originally presented by \cite{perez_98}; additional computational
details have recently appeared in \cite{fouskakis_ntzoufras_2012}.

Suppose we have an imaginary/training data set $\dn{ y }^*$, of size $n^*$,
and design matrix $\mt{ X }^*$ of size $n^* \times ( p + 1 ) \,$, where $p$
denotes the total number of available covariates. Then the EPP distribution,
given by (\ref{epp}), will depend on $\mt{ X }^*$ but not on $\dn{ y }^*$,
since the latter is integrated out. The selection of a \textit{minimal
training sample} has been proposed, to make the information content of the
prior as small as possible, and this is an appealing idea.  However, even
the definition of \textit{minimal} turns out to be open to question, since
it is problem-specific (which models are we comparing?) and data-specific
(how many variables are we considering?). One possibility is to specify the
size of the minimal training sample either from (a) the dimension of the
full model or (b) the dimension of the larger model in every pairwise model
comparison performed. But, as will be seen below, when $n$ is not much
larger than $p$, working with a minimal training sample can result in a
prior that is far more influential than intended. Additionally, if the data
derive from a highly structured situation, such as a randomized complete
block experiment, most choices of a small part of the data to act as a
training sample would be untypical.

Even if the minimal-training-sample idea is accepted, the problem of
choosing such a subset of the full data set still remains. A natural
solution involves computing the arithmetic mean (or some other summary of
distributional center) of the Bayes factors over all possible training
samples, but this approach can be computationally infeasible, especially
when $n$ is much larger than $p$; for example, with $( n, p ) = ( 100, 50 )$
and $( 500, 100 )$ there are about $10^{ 29 }$ and $10^{ 107 }$ possible
training samples, respectively, over which to average. An obvious choice at
this point is to take a random sample from the set of all possible minimal
training samples, but this adds an extraneous layer of Monte-Carlo noise to
the model-comparison process. These difficulties have been well-documented
in the literature, but the quest for a fully satisfactory solution is still
on-going; for example, \cite{berger_pericchi_2004} note that they ``were
unable to define any type of `optimal' training sample."

An approach to choosing covariate values for the training sample has been
proposed by researchers working with intrinsic priors
(\cite{casella_moreno_2006}; \cite{giron_etal_2006a};
\cite{moreno_giron_2008}; \cite{casella_etal_2009}), since the same problem
arises there too. They consider all pairwise model comparisons, either
between the full model and each nested model, or between every model
configuration and the null model, or between two nested models. They used
training samples of size defined by the dimension of the full model in the
first case, or by the dimension of the larger model in every pairwise
comparison in the second and third cases. In all three settings, they
proved that the intrinsic prior of the parameters of the larger model in
each pairwise comparison, denoted here by $M_k$, depends on the imaginary
covariate values only through the expression $\mt{ W }_k^{ -1 } = ( \mt{
X_k }^{ *^T } \mt{ X_k }^* )^{ -1 }$, where $\mt{ X }_k^*$ is the imaginary
design matrix of dimension $( d_k + 1 ) \times d_k \,$ for a minimal
training sample of size $( d_k + 1 )$. Then, driven by the idea of the
arithmetic intrinsic Bayes factor, they avoid the dependence on the
training sample by replacing $\mt{ W }_k^{ -1 }$ with its average over all
possible training samples of minimal size. This average can be proved to be
equal to $\frac{ n }{ d_k + 1 } \left( \mt{ X }_k^T \mt{ X }_k \right)^{ -1
}$, where $\mt{ X }_k$ is the design matrix of the larger model in each
pairwise comparison, and therefore no subsampling from the $\mt{ X }_k$
matrix is needed.

Although this approach seems intuitively sensible and dispenses with the
extraction of the submatrices from $\mt{ X }_k$, it is unclear if the
procedure retains its intrinsic interpretation, i.e., whether it is
equivalent to the arithmetic intrinsic Bayes factor. Furthermore, and more
seriously, the resulting prior can be influential when $n$ is not much
larger than $p$, in contrast to the prior we propose here, which has a
unit-information interpretation.

\section{Power-expected-posterior (PEP) priors}

\label{sec_pep}

In this paper, starting with the EPP methodology, we combine ideas from the
power-prior approach of \cite{ibrahim_chen_2000} and the
unit-information-prior approach of \cite{kass_wasserman_95}. As a first
step, the likelihoods involved in the EPP distribution are raised to the
power $\frac{ 1 }{ \delta }$ and density-normalized. Then we set the power
parameter $\delta$ equal to $n^*$, to represent information equal to one
data point; in this way the prior corresponds to a sample of size one with
the same sufficient statistics as the observed data. Regarding the size of
the training sample, $n^*$, this could be any integer from $( p + 2 )$ (the
minimal training sample size) to $n$. As will become clear below, we have
found that significant advantages (and no disadvantages) arise from the
choice $n^* = n$, from which $\mt{ X }^* = \mt{ X }$. In this way we
completely avoid the selection of a training sample and its effects on the
posterior model comparison, while still holding the prior information
content at one data point. Sensitivity analysis for different choices of
$n^*$ is performed as part of the first set of experimental results below
(see Section \ref{sim}).

For any $M_\ell \in { \cal M }$, we denote by $\pi_\ell^N ( \dn{ \beta
}_\ell, \sigma_\ell^2 | \bx_\ell^* )$ the baseline prior for model
parameters $\dn{ \beta }_\ell$ and $\sigma_\ell^2$. Then the
\textit{power-expected-posterior} (PEP) prior $\pi_\ell^{ PEP } ( \dn{ \beta
}_\ell, \sigma_\ell^2 | \, \bx_\ell^* \, , \delta )$ takes the following
form:
\begin{equation} \label{pep1} 
\pi_\ell^{ PEP } ( \dn{ \beta }_\ell, \sigma_\ell^2 \, | \, \bx_\ell^* \, ,
\delta ) =  \int \pi_\ell^N ( \dn{ \beta }_\ell, \sigma_\ell^2 \, | \dn{ y
}^*, \delta ) \, m_0^N ( \dn{ y }^* | \, \bx_0^* \, , \delta ) \, d \dn{ y
}^* \, ,
\end{equation}
where 
\begin{equation} \label{prior-post} 
\pi_\ell^N ( \dn{ \beta }_\ell, \sigma_\ell^2 \, | \dn{ y }^*, \delta ) =
\frac{ f ( \dn{ y }^* | \, \dn{ \beta }_\ell \, , \sigma_\ell^2, M_\ell \,
; \mt{ X }_\ell^* \, , \delta ) \pi_\ell^N ( \dn{ \beta }_\ell,
\sigma_\ell^2 | \bx_\ell^* ) }{ m_\ell^N ( \dn{ y }^* | \, \bx_\ell^*\, ,
\delta ) },
\end{equation}
and $f ( \dn{ y }^* | \, \dn{ \beta }_\ell \, , \sigma_\ell^2, M_\ell \, ;
\mt{ X }_\ell^* \, , \delta ) \propto f ( \dn{ y }^* | \dn{ \beta }_\ell \,
, \sigma_\ell^2, M_\ell \, ; \mt{ X }_\ell^* )^{ \frac{ 1 }{ \delta } }$ is
the EPP likelihood raised to the power $\frac{ 1 }{ \delta }$ and
density-normalized, i.e.,
\begin{eqnarray} \label{power_likelihood}
f ( \dn{ y }^* | \, \dn{ \beta }_\ell \, , \sigma_\ell^2, M_\ell \, ; \mt{
X }_\ell^* \, , \delta ) & = & \frac{ f ( \dn{ y }^* | \dn{ \beta }_\ell,
\sigma_\ell^2, M_\ell \, ; \mt{ X }_\ell^* )^{ \frac{ 1 }{ \delta } } }{
\int f ( \dn{ y }^* | \dn{ \beta }_\ell, \sigma_\ell^2, M_\ell \, ; \mt{ X
}_\ell^* )^{ \frac{ 1 }{ \delta } } d \dn{ y }^*} = \frac{ f_{ N_{ n^* } }
( \dn{ y }^* \, ; \, \mt{ X }_\ell^* \dn{ \beta }_\ell \, , \sigma_\ell^2
\, \mt{ I }_{ n^* } )^{ \frac{ 1 }{ \delta } } }{ \int f_{ N_{ n^* } } (
\dn{ y }^* \, ; \, \mt{ X }_\ell^* \dn{ \beta }_\ell \, , \sigma_\ell^2 \,
\mt{ I }_{ n^* } )^{ \frac{ 1 }{ \delta } } d \dn{ y }^* } \nonumber \\
& = & f_{ N_{ n^* } } ( \dn{ y }^* \, ; \, \mt{ X }_\ell^* \dn{ \beta
}_\ell \, , \delta \, \sigma_\ell^2 \mt{ I }_{ n^* } ) \, ;
\end{eqnarray}
here $f_{ N_d } ( \dn{ y } \, ; \, \dn{ \mu }, \dn{ \Sigma } )$ is the
density of the $d$-dimensional Normal distribution with mean $\dn{ \mu }$
and covariance matrix $\dn{ \Sigma }$, evaluated at $\dn{ y }$.

The distribution $m_\ell^N ( \dn{ y }^* | \, \bx_\ell^* \, , \delta )$
appearing in (\ref{pep1}) (for $\ell = 0$) and (\ref{prior-post}) is the
prior predictive distribution (or the marginal likelihood), evaluated at
$\dn{ y }^*$, of model $M_\ell$ with the power likelihood defined in
(\ref{power_likelihood}) under the baseline prior $\pi^N_\ell ( \dn{ \beta
}_\ell, \sigma_\ell^2 \, | \, \bx_\ell^* )$, i.e.,
\begin{equation} \label{new2-3}
m_\ell^N ( \dn{ y }^* | \, \bx_\ell^* \, , \delta ) = \int \! \! \int f_{
N_{ n^* } } ( \dn{ y }^* \, ; \, \mt{ X }_\ell^* \, \dn{ \beta }_\ell \, ,
\delta \, \sigma_\ell^2 \, \mt{ I }_{ n^* } ) \, \pi^N_\ell ( \dn{ \beta
}_\ell \, , \sigma_\ell^2 \, | \, \bx_\ell^* ) \, d \dn{ \beta }_\ell \, d
\sigma_\ell^2 \, .
\end{equation}
From (\ref{pep1}) and (\ref{prior-post}), the PEP prior can be re-written
as 
\begin{equation} \label{pep}
\pi_\ell^{ PEP } ( \dn{ \beta }_\ell, \sigma_\ell^2 \, | \, \bx_\ell^* \, ,
\delta ) = \pi_\ell^N ( \dn{ \beta }_\ell, \sigma_\ell^2 | \bx_\ell^* )
\int \frac{ m_0^N ( \dn{ y }^* | \, \bx_0^* \, , \delta ) }{ m_\ell^N (
\dn{ y }^* | \, \bx_\ell^* \, , \delta ) } \, f ( \dn{ y }^* | \, \dn{
\beta }_\ell \, , \sigma_\ell^2, M_\ell \, ; \mt{ X }_\ell^* \, , \delta )
\, d \dn{ y }^* \, .
\end{equation}
Under the PEP prior distribution (\ref{pep}), the posterior distribution of
the model parameters $( \dn{ \beta }_\ell \, , \sigma_\ell^2) $ is
\begin{eqnarray} \label{pep_post_gen}
\pi_\ell^{ PEP } ( \dn{ \beta }_\ell, \sigma_\ell^2 | \dn{ y } ; \mt{ X
}_\ell, \mt{ X }_\ell^*, \delta ) & \propto & \int \pi_\ell^{ N } ( \dn{
\beta }_\ell, \sigma_\ell^2 | \dn{ y }, \dn{ y }^* ; \mt{ X }_\ell, \mt{ X
}_\ell^*, \delta ) \times \nonumber \\
& & \hspace*{0.25in} \, m_\ell^N ( \dn{ y } | \dn{ y }^* ; \bx_\ell,
\bx_\ell^* \, , \delta ) \, m_0^N ( \dn{ y }^* | \bx_0^* \, , \delta ) \, d
\dn{ y }^*,
\end{eqnarray}
where $\pi_\ell^{ N } ( \dn{ \beta }_\ell, \sigma_\ell^2 | \dn{ y }, \dn{ y
}^* ; \mt{ X }_\ell, \mt{ X }_\ell^*, \delta )$ and $m_\ell^N ( \dn{ y } |
\dn{ y }^* ; \bx_\ell, \bx_\ell^* \, , \delta )$ are the posterior
distribution of $( \dn{ \beta }_\ell, \sigma_\ell^2 )$ and the marginal
likelihood of model $M_\ell$, respectively, using data $\dn{ y }$ and design
matrix $\mt{ X }_\ell$ under prior $\pi_\ell^{ N } ( \dn{ \beta }_\ell \, ,
\sigma_\ell^2 | \dn{ y }^*; \, \bx_\ell^* \, , \delta )$ --- i.e., the
posterior of $( \dn{ \beta}_\ell \, , \sigma_\ell^2 )$ with power Normal
likelihood (\ref{power_likelihood}) and baseline prior $\pi_\ell^N ( \dn{
\beta }_\ell, \sigma_\ell^2 | \bx_\ell^* )$.

In what follows we present results for the PEP prior using two specific
baseline prior choices: the independence Jeffreys prior (improper) and the
$g$-prior (proper). The first is the usual choice among researchers
developing objective variable-selection methods, but the posterior results
using this first baseline-prior choice can also be obtained as a limiting
case of the results using the second baseline prior (see Section
\ref{connection_j_z}); usage of this second approach can lead to
significant computational acceleration with the Jeffreys baseline prior.

\subsection{PEP-prior methodology with the Jeffreys baseline prior: J-PEP}

\label{pep-jeffreys}

Here we use the independence Jeffreys prior (or reference prior) as the
baseline prior distribution. Hence for $M_\ell \in {\cal M}$ we have
\begin{equation} \label{JE}
\pi_\ell^N ( \dn{ \beta }_\ell \, , \sigma^2 \, | \, \bx_\ell^* ) = \frac{
c_\ell }{ \sigma_\ell^2 } \, ,
\end{equation}
where $c_\ell$ is an unknown normalizing constant; we refer to the resulting
PEP prior as \textit{J-PEP}.

\subsubsection{Prior setup}

\label{pep-jeffreys-prior-setup}

Following (\ref{pep}) for the baseline prior (\ref{JE}) and the power
likelihood specified in (\ref{power_likelihood}), the PEP prior, for any
model $M_\ell \,$, now becomes
\begin{eqnarray} \label{pep_j}
\pi_\ell^{ \textit{J-PEP} } ( \dn{ \beta }_\ell, \sigma_\ell^2 | \bx_\ell^*
\, , \delta) & = & \int f_{ N_{ d_\ell } } \! \left[ \dn{ \beta }_\ell \, ;
\, \widehat{ \dn{ \beta } }_\ell^*,  \delta \, ( \mt{ X }_\ell^{ *^T } \mt{
X }_\ell^* )^{ -1 } \sigma_\ell^2 \right]  \times \nonumber \\ & &
\hspace*{0.25in} f_{ IG } \Big( \sigma_\ell^2 \, ; \, \frac{ n^* - d_\ell }{
2 }, \frac{ RSS_\ell^* }{ 2 \delta } \Big) \, m_0^N ( \dn{ y }^* | \bx_0^*
\, , \delta ) \, d \dn{ y }^* \, ,
\end{eqnarray}
where $f_{ IG } \left( y \, ; \, a, b \right)$ is the density of the
Inverse-Gamma distribution with parameters $a$ and $b$ and mean $\frac{ b }{
a - 1 }$, evaluated at $y$. Here $\widehat{ \dn{ \beta } }_\ell^* = ( \mt{ X
}_\ell^{ *^T } \mt{ X }_\ell^* )^{ -1 } \mt{ X }_\ell^{ *^T } \dn{ y }^*$ is
the MLE with outcome vector $\dn{ y }^*$ and design matrix $\mt{ X
}_\ell^*$, and $RSS_\ell^* = \dn{ y }{ ^*{ ^T } } \big[ \identst - \mt{ X
}_\ell^* ( \mt{ X }_\ell{ ^*{ ^T } } \mt{ X }_\ell^* )^{ -1 } \mt{ X }_\ell{
^*{ ^T } } \big] \dn{ y }^* $ is the residual sum of squares using $( \dn{ y
}^*, \mt{ X }_\ell^*)$ as data. The prior predictive distribution of any
model $M_\ell$ with power likelihood defined in (\ref{power_likelihood})
under the baseline prior (\ref{JE}) is given by
\begin{equation} \label{prior_predictive_j}
m_\ell^N ( \dn{ y^* } \, | \, \bx_\ell^* \, , \delta ) =  c_\ell \, \pi^{
\frac{ 1 }{ 2 } ( d_\ell - n^* ) } \, | \mt{ X }_\ell^{ *^T } \mt{ X
}_\ell^* |^{ - \frac{ 1 }{ 2 } } \, \Gamma \left( \frac{ n^* - d_\ell }{ 2 }
\right) RSS_\ell^{ *^{ - \left( \tfrac{ n^* - d_\ell }{ 2 } \right) } } \, .
\end{equation}

\subsubsection{Posterior distribution}

\label{pep-jeffreys-posterior}

For the PEP prior (\ref{pep_j}), the posterior distribution of the model
parameters $( \dn{ \beta }_\ell \, , \sigma_\ell^2 )$ is given by
(\ref{pep_post_gen}) with $f ( \dn{ \beta }_\ell, \sigma_\ell^2 | \dn{ y },
\dn{ y }^*, M_\ell \, ; \mt{ X }_\ell, \mt{ X }_\ell^*, \delta )$ and
$m_\ell^N ( \dn{ y } | \dn{ y }^* ; \bx_\ell, \bx_\ell^* \, , \delta )$ as
the posterior distribution of $( \dn{ \beta }_\ell, \sigma_\ell^2 )$ and
the marginal likelihood of model $M_\ell$, respectively, using data $\dn{ y
}$, design matrix $\mt{ X }_\ell$, and the Normal-Inverse-Gamma
distribution appearing in (\ref{pep_j}) as prior.  Hence
\begin{eqnarray} \label{new2-8}
\pi_\ell^N ( \dn{ \beta }_\ell | \sigma_\ell^2 , \dn{ y }, \dn{ y }^*;
\mt{ X }_\ell, \mt{ X }_\ell^*, \delta ) & = & f_{ N_{ d_\ell } } \big(
\dn{ \beta }_\ell \, ; \, \widetilde{ \dn{ \beta } }^N, \; \widetilde{
\mt{ \Sigma } }^N \sigma_\ell^2 \big) \ \mbox{and} \nonumber \\
\pi_\ell^N ( \sigma_\ell^2 | \dn{ y }, \dn{ y }^*; \mt{ X }_\ell, \mt{
X}_\ell^*, \delta ) & = & f_{ IG } ( \sigma_\ell^2 \, ; \, \widetilde{ a
}_\ell^N, \widetilde{ b }_\ell^N ) \, ,
\end{eqnarray}
with
\begin{eqnarray} \label{post_params_j}
\widetilde{ \dn{ \beta } }^N & = & \widetilde{ \mt{ \Sigma } }^N ( \mt{ X
}_\ell^T \dn{ y } + \delta^{ -1 } \mt{ X }_\ell^{ *^T } \dn{ y }^* ), \
\widetilde{ \mt{ \Sigma } }^N = \left[ \mt{ X }_\ell^T \mt{ X }_\ell + \,
\delta^{ -1 } \mt{ X }_\ell^{ *^T } \mt{ X }_\ell^* \right]^{ -1 } \, \
\mbox{and} \nonumber \\
\widetilde{ a }_\ell^N & = & \frac{ n + n^* -d_\ell }{ 2 } \, , \
\widetilde{ b }_\ell^N = \frac{ SS_\ell^N + \delta^{ -1 } RSS_\ell^* }{ 2 }
\, .
\end{eqnarray}
Here
\begin{eqnarray} \label{ss-ell-N-1}
SS_\ell^N & = & \big( \dn{ y } - \mt{ X }_\ell \, \widehat{ \dn{ \beta }
}_\ell^* )^T \left[ \ident + \delta \, \mt{ X }_\ell ( \mt{ X }_\ell^{ *^T
} \mt{ X }_\ell^* )^{ -1 } \mt{ X }_\ell^T \right]^{ -1 } \big( \dn{ y } -
\mt{ X }_\ell \, \widehat{ \dn{ \beta } }_\ell^* )
\end{eqnarray}
and
\begin{equation} \label{posterior_predictive_j}
m_\ell^N ( \dn{ y } | \dn{ y }^* ; \bx_\ell, \bx_\ell^* \, , \delta ) = f_{
St_n } \left\{ \dn{ y } \, ; \, n^* - d_\ell, \, \mt{ X }_\ell \widehat{
\dn{ \beta } }_\ell^*, \, \frac{ RSS^*_\ell }{ \delta( n^* - d_\ell ) }
\left[ \mt{ I }_{ n } + \delta \, \mt{ X }_\ell ( \mt{ X }_\ell^{ *^T }
\mt{ X }_\ell^* )^{ -1 } \mt{ X }_\ell^T \right] \right\} \, ,
\end{equation}
in which $St_n ( \cdot \, ; d, \dn{ \mu }, \Sigma )$ is the multivariate
Student distribution in $n$ dimensions with $d$ degrees of freedom,
location $\dn{ \mu }$ and scale $\Sigma$. Thus the posterior distribution
of the model parameters $( \dn{ \beta }_\ell \, ,\sigma_\ell^2 )$ under the
PEP prior (\ref{pep_j}) is
\begin{eqnarray} \label{posterior_pep_nig}
\pi_\ell^{ \textit{J-PEP} } ( \dn{ \beta }_\ell, \sigma_\ell^2 | \dn{ y } ;
\mt{ X }_\ell, \mt{ X }_\ell^*, \delta ) & \propto & \int f_{ N_{ d_\ell } }
\big( \dn{ \beta }_\ell \, ; \, \widetilde{ \dn{ \beta } }^N, \; \widetilde{
\mt{ \Sigma } }^N \sigma_\ell^2 \big) \, f_{ IG } ( \sigma_\ell^2 \, ; \,
\widetilde{ a }_\ell^N, \widetilde{ b }_\ell^N ) \times \nonumber \\
& & \hspace*{0.25in} m_\ell^N ( \dn{ y } | \dn{ y }^* ;
\, \bx_\ell, \bx_\ell^* \, , \delta ) \, m_0^N ( \dn{ y }^* | \bx_0^* \, ,
\delta ) \, d \dn{ y }^* \, ,
\end{eqnarray}
with $m_0^N ( \dn{ y }^* | \bx_0^* \, , \delta )$ given in
(\ref{prior_predictive_j}). A detailed MCMC scheme for sampling from this
distribution is presented in Section 2 of the web Appendix.

\subsection{PEP-prior methodology with the $g$-prior as baseline: Z-PEP}

\label{pep-zellner}

Here we use the Zellner $g$-prior as the baseline prior distribution; in
other words, for any $M_\ell \in { \cal M }$
\begin{equation} \label{zellners-n-star-prior}
\pi_\ell^N ( \dn{ \beta }_\ell | \sigma_\ell^2 \, ; \, \bx_\ell^* ) = f_{
N_{ d_\ell } } \left[ \dn{ \beta }_\ell \, ; \, \dn{ 0 }, g \, ( \bx_\ell^{
*^T } \bx_\ell^* )^{ -1 } \sigma_\ell^2 \right] \textrm{ and } \pi_\ell^N (
\sigma_\ell^2 ) = f_{ IG } \left( \sigma_\ell^2 \, ; \, a_\ell, b_\ell
\right) \, .
\end{equation}
We refer to the resulting PEP prior as \textit{Z-PEP}. Note that the usual
improper reference prior for $\sigma_\ell$ could easily be used instead,
but for computational reasons we prefer here to use the Inverse-Gamma prior
(recall that for $a_\ell$ and $b_\ell$ approximately equal to zero, the
Inverse-Gamma prior degenerates to the improper reference prior).

\subsubsection{Prior setup}

\label{pep-zellner-prior}

For any model $M_\ell$, under the baseline prior setup
(\ref{zellners-n-star-prior}) and the power likelihood
(\ref{power_likelihood}), the prior predictive distribution is
\begin{equation} \label{prior_predictive_z}
m_\ell^N ( \dn{ y^* } \, | \, \bx_\ell^* \, , \delta ) = f_{ St_{ n^* } }
\Big( \dn{ y }^* \, ; 2 \, a_\ell, \dn{ 0 }, \frac{ b_\ell }{ a_\ell } {
\Lambda_\ell^* }^{ -1 } \Big) \, ,
\end{equation}
where
\begin{equation} \label{lambda}
{ \Lambda_\ell^* }^{ -1 } = \delta \left[ \mt{ I }_{ n^* } - \frac{ g }{ g
+ \delta } \mt{ X }_\ell^* \left( { \mt{ X }_\ell^* }^T \mt{ X }_\ell^*
\right)^{ -1 }{ \mt{ X }_\ell^* }^T \right]^{ -1 } = \delta \, \mt{ I }_{
n^* } + g \, \mt{ X }_\ell^* \left( { \mt{ X }_\ell^* }^T \mt{ X }_\ell^*
\right)^{ -1 }{ \mt{ X }_\ell^* }^T \, .
\end{equation}
In the special case of the constant model, (\ref{lambda}) simplifies to
$\left( \delta \, \mt{ I }_{ n^* } + \frac{ g }{ n } \dn{ 1 }_{ n^* } \dn{
1 }_{ n^* }^T \right)$, where $\dn{ 1 }_{ n^* }$ is a vector of length
$n^*$ with all elements equal to one.

Following (\ref{pep}) for the baseline prior (\ref{zellners-n-star-prior})
and the power likelihood specified in (\ref{power_likelihood}), the Z-PEP
prior, for any model $M_\ell \,$, now becomes
\begin{eqnarray} \label{pep_z}
\pi_\ell^{ \textit{Z-PEP} } ( \dn{ \beta }_\ell, \sigma_\ell^2 | \bx_\ell^*
\, , \delta ) & = & \int f_{ N_{ d_\ell } } \left[ \dn{ \beta }_\ell \, ; \,
w \, \widehat{ \dn{ \beta } }_\ell^*, w \, \delta \, ( \mt{ X }_\ell^{ *^T }
\mt{ X }_\ell^* )^{ -1 } \sigma_\ell^2 \right] \times \nonumber \\
& & \hspace*{0.25in} f_{ IG } \Big( \sigma_\ell^2 \, ; \, a_\ell + \frac{
n^* }{ 2 }, b_\ell + \frac{ SS_\ell^* }{ 2 } \Big) \, m_0^N ( \dn{ y }^* |
\bx_0^* \, , \delta ) \, d \dn{ y }^* \, .
\end{eqnarray}
Here $w = \frac{ g }{ g + \delta }$ is the shrinkage weight, $\widehat{ \dn{
\beta } }_\ell^* = ( \mt{ X }_\ell^{ *^T } \mt{ X }_\ell^*)^{ -1 } \mt{ X
}_\ell^{ *^T } \dn{ y }^*$ is the MLE with outcome vector $\dn{ y}^*$ and
design matrix $\mt{ X }_\ell^*$, and $SS_\ell^* = \dn{ y }^{ *^T }
\Lambda_\ell^* \, \dn{ y }^*$ is the posterior sum of squares.

The prior mean vector and covariance matrix of $\dn{ \beta }_\ell$, and the
prior mean and variance of $\sigma_\ell^2$, can be calculated analytically
from these expressions; details are available in Theorems 1 and 2 in Section
1 of the web Appendix.

\subsubsection{Posterior distribution}

\label{sec_posterior_pep}

The distributions $\pi_\ell^{ N } ( \dn{ \beta }_\ell, \sigma_\ell^2 | \dn{
y }, \dn{ y }^* ; \mt{ X }_\ell, \mt{ X }_\ell^*, \delta )$ and $m_\ell^N (
\dn{ y } | \dn{ y }^* ; \bx_\ell, \bx_\ell^* \, , \delta )$ involved in the
calculation of the posterior distribution (\ref{pep_post_gen}) are now the
posterior distribution of $( \dn{ \beta }_\ell, \sigma_\ell^2 )$ and the
marginal likelihood of model $M_\ell$, respectively, using data $\dn{ y }$,
design matrix $\mt{ X }_\ell$, and $\pi_\ell^{ N } ( \dn{ \beta }_\ell \, ,
\sigma_\ell^2 | \dn{ y }^*; \, \bx_\ell^* \, , \delta )$ as a prior density
(which is the Normal-Inverse-Gamma distribution appearing in
(\ref{pep_z})).  Therefore the posterior distribution of the model
parameters $( \dn{ \beta }_\ell \, ,\sigma_\ell^2 )$ under the Z-PEP prior
(\ref{pep_z}) is given by
\begin{eqnarray} \label{posterior_pep_z}
\pi_\ell^{ \textit{Z-PEP} } ( \dn{ \beta }_\ell, \sigma_\ell^2 | \dn{ y }
; \mt{ X }_\ell, \mt{ X }_\ell^*, \delta ) & \propto & \int f_{ N_{ d_\ell }
} \big( \dn{ \beta }_\ell \, ; \, \widetilde{ \dn{ \beta } }^N, \;
\widetilde{ \mt{ \Sigma } }^N \sigma_\ell^2 \big) \, f_{ IG } (
\sigma_\ell^2 \, ; \, \widetilde{ a }_\ell^N, \widetilde{ b }_\ell^N )
\times \nonumber \\
& & \hspace*{0.25in} m_\ell^N ( \dn{ y } | \dn{ y }^* ;
\,\bx_\ell, \bx_\ell^* \, , \delta ) \, m_0^N ( \dn{ y }^* | \bx_0^* \, ,
\delta ) \, d \dn{ y }^* \, ,
\end{eqnarray}
with
\begin{eqnarray} \label{post_params1_z}
\widetilde{ \dn{ \beta } }^N & = & \widetilde{ \mt{ \Sigma } }^N ( \mt{ X
}_\ell^T \dn{ y } + \delta^{ -1 } \mt{ X }_\ell^{ *^T } \dn{ y }^* ), \
\widetilde{ \mt{ \Sigma } }^N = \left[ \mt{ X }_\ell^T \mt{ X }_\ell + ( w
\, \delta )^{ -1 } \mt{ X }_\ell^{ *^T } \mt{ X }_\ell^* \right]^{ -1 } \,
\ \textrm{and} \nonumber \\
\widetilde{ a }_\ell^N & = & \frac{ n + n^* }{ 2 } + a_\ell \, , \
\widetilde{ b }_\ell^N = \frac{ SS_\ell^N + SS_\ell^* }{ 2 } + b_\ell \, .
\end{eqnarray}
Here
\begin{equation} \label{ss-ell-n-1}
SS_\ell^N = \big( \dn{ y } - w \, \mt{ X }_\ell \, \widehat{ \dn{ \beta }
}_\ell^* )^T \left[ \ident + \delta \, w \, \mt{ X }_\ell ( \mt{ X }_\ell^{
*^T } \mt{ X }_\ell^* )^{ -1 } \mt{ X }_\ell^T \right]^{ -1 } \big( \dn{ y
} - w \, \mt{ X }_\ell \, \widehat{ \dn{ \beta } }_\ell^* ) \, ,
\end{equation}
while
\begin{equation} \label{posterior_predictive_z}
m_\ell^N ( \dn{ y } | \dn{ y }^* ; \bx_\ell, \bx_\ell^* \, , \delta ) =
f_{ St_n } \left\{ \dn{ y } \, ; \, 2 \, a_\ell + n^*, w \, \mt{ X }_\ell
\widehat{ \dn{ \beta } }_\ell^*, \, \frac{ 2 b_\ell + SS^*_\ell }{ 2 \,
a_\ell + n^* } \left[ \mt{ I }_{ n } + w \, \delta \, \mt{ X }_\ell ( \mt{
X }_\ell^{ *^T } \mt{ X }_\ell^* )^{ -1 } \mt{ X }_\ell^T \right] \right\}
\, ,
\end{equation}
and $m_0^N ( \dn{ y }^* | \bx_0^* \, , \delta )$ is given in
(\ref{prior_predictive_z}). A detailed MCMC scheme for sampling from this
distribution is presented in Section 2 of the web Appendix.

\subsubsection{Specification of hyper-parameters}

\label{sec_prior_parameters}

The marginal likelihood for the Z-PEP prior methodology, using the
$g$-prior as a baseline, depends on the selection of the hyper-parameters
$g$, $a_\ell$ and $b_\ell$. We make the following proposals for specifying
these quantities, in settings in which strong prior information about the
parameter vectors in the models is not available.

The parameter $g$ in the Normal baseline prior is set to $\delta \, n^*$,
so that with $\delta = n^*$ we use $g = ( n^* )^2$. This choice will make
the $g$-prior contribute information equal to one data point within the
posterior $\pi_\ell^{ N } ( \dn{ \beta }_\ell \, , \sigma_\ell^2 | \dn{ y
}^*; \, \bx_\ell^* \, , \delta )$. In this manner, the entire Z-PEP prior
contributes information equal to $\left( 1 + \frac{ 1 }{ \delta } \right)$
data points.

We set the parameters $a_\ell$ and $b_\ell$ in the Inverse-Gamma baseline
prior to 0.01, yielding a baseline prior mean of 1 and variance of 100
(i.e., a large amount of prior uncertainty) for the precision parameter;
our method yields similar results across a broad range of small values of
$a_\ell$ and $b_\ell$. (If strong prior information about the model
parameters is available, Theorems 1 and 2 in Section 1 of the web Appendix
can be used to guide the choice of $a_\ell$ and $b_\ell$.) 

\subsection{Connection between the J-PEP and Z-PEP distributions}

\label{connection_j_z}

By comparing the posterior distributions under the two different baseline
schemes described in Sections \ref{pep-jeffreys} and \ref{pep-zellner}, it
is straightforward to prove that they coincide under the following
conditions $( * )$: large $g$ (and therefore $w \approx 1$), $a_\ell = -
\frac{ d_\ell }{ 2 }$ and $b_\ell = 0$.

To be more specific, the posterior distribution in both cases takes the form
of equation (\ref{posterior_pep_nig}). The parameters of the
Normal-Inverse-Gamma distribution (see equations (\ref{post_params1_z}))
involved in the posterior distribution using the $g$-prior as baseline
become equal to the corresponding parameters for the Jeffreys baseline (see
equations (\ref{post_params_j})) with parameter values $( * )$. Similarly,
the conditional marginal likelihood $m_\ell^N ( \dn{ y } | \dn{ y }^* ;
\,\bx_\ell, \bx_\ell^* \, , \delta )$ under the two baseline priors (see
equations (\ref{posterior_predictive_j}) and (\ref{posterior_predictive_z}))
becomes the same under conditions $( * )$.

Finally, the prior predictive densities $m_0^N ( \dn{ y }^* | \bx_0^* \, ,
\delta )$ involved in equations (\ref{posterior_pep_nig}) and
(\ref{posterior_pep_z}) can be written as $m_0^N ( \dn{ y }^* | \bx_0^* \,
, \delta ) \propto (2 \, b_\ell \, + \, SS_\ell^* )^{ - \frac{ n^* + a_\ell
}{ 2 } }$ for the $g$-prior baseline and as $m_0^N ( \dn{ y }^* | \bx_0^*
\, , \delta ) \propto RSS_\ell^{ *^{ - \frac{ n^* - d_\ell }{ 2 } } }$ for
the Jeffreys baseline. For large values of $g$, $SS_\ell^* \rightarrow
\delta^{ -1 } RSS_\ell^*$, and the two un-normalized prior predictive
densities clearly become equal if we further set $a_\ell = - \frac{ d_\ell
}{ 2 }$ and $b_\ell = 0$. Any differences in the normalizing constants of
$m_0^N ( \dn{ y }^* | \bx_0^* \, , \delta )$ cancel out when normalizing
the posterior distributions (\ref{posterior_pep_nig}) and
(\ref{posterior_pep_z}).

For these reasons, the posterior results using the Jeffreys prior as
baseline can be obtained as a special (limiting) case of the results using
the $g$-prior as baseline. This can be beneficial for the computation of
the posterior distribution, which is detailed in Section 2 of the web
Appendix, and for the estimation of the marginal likelihood presented in
Section \ref{marginal_likelihood_pep_gen}.

\section{Marginal-likelihood computation}

\label{marginal_likelihood_pep_gen}

Under the PEP-prior approach, it is straightforward to show that the
marginal likelihood of any model $M_\ell \in { \cal M }$ is
\begin{equation} \label{new2-17}
m^{ PEP }_\ell ( \dn{ y } | \mt{ X }_\ell \, , \mt{ X }_\ell^* \,
, \delta ) = m_\ell^N ( \dn{ y } | \bx_\ell \, , \bx_\ell^* \,) \int \frac{
m_\ell^N ( \dn{ y }^* | \dn{ y } , \bx_\ell, \bx_\ell^* \, , \delta ) }{
m_\ell^N ( \dn{ y }^* | \bx_\ell^* \, , \delta ) } \, m_0^N ( \dn{ y }^* |
\bx_0^* \, , \delta ) \, d \dn{ y }^* \, .
\end{equation}
Note that in the above expression $ m_\ell^N ( \dn{ y } | \bx_\ell \, ,
\mt{ X }_\ell^* )$ is the marginal likelihood of model $M_\ell$ for the
actual data under the baseline prior and therefore, under the baseline
$g$-prior (\ref{zellners-n-star-prior}), is given by
\begin{equation} \label{marginal-like-zellners}
m_\ell^N ( \dn{ y } | \bx_\ell \, , \bx_\ell^* ) = f_{ St_{ n } } \left\{
\dn{ y } \, ; 2 \, a_\ell, \dn{ 0 }, \frac{ b_\ell }{ a_\ell } \left[ \mt{
I }_{ n } + g \, \mt{ X }_\ell \left( \bx_\ell^{ *^T } \bx_\ell^* \right)^{
-1 }{ \mt{ X }_\ell }^T \right] \right\} \, ;
\end{equation}
under the Jeffreys baseline prior (\ref{JE}), $m_\ell^N ( \dn{ y } |
\bx_\ell \, , \bx_\ell^* )$ is given by equation (\ref{prior_predictive_j})
with data $(\dn{y}, \mt{ X }_\ell )$.

In settings in which the marginal likelihood (\ref{new2-17}) is not
analytically tractable, we have obtained four possible Monte-Carlo
estimates. In Section \ref{pi-g-prior-1} we show that two of these
possibilities are far less accurate than the other two; we detail the less
successful approaches in Section 3 of the web Appendix. The other two
(more accurate) methods are as follows:

\begin{enumerate}

\item[(1)]

Generate $\dn{ y }^{ * ( t ) } \, ( t = 1, \dots, T )$ from $m_\ell^N (
\dn{ y }^* | \dn{ y }, \, \bx_\ell, \bx_\ell^* \, , \delta )$ and estimate
the marginal likelihood by
\begin{equation} \label{marginal_like_estim3}
\hat{ m }^{ PEP }_\ell ( \dn{ y } | \mt{ X }_\ell \, , \mt{ X }_\ell^*
\, , \delta ) = m_\ell^N ( \dn{ y } | \bx_\ell, \bx_\ell^* ) \left[ \frac{
1 }{ T } \sum_{ t = 1}^T \frac{ m_0^N ( \dn{ y }^{ * ( t ) } | \bx_0^* \, ,
\delta ) }{ m_\ell^N ( \dn{ y }^{ * ( t ) } | \bx_\ell^* \, , \delta ) }
\right] \, .
\end{equation}

\item[(2)]

Generate $\dn{ y }^{ * ( t ) } \, ( t = 1, \dots, T )$ from $m_\ell^N (
\dn{ y }^* | \, \dn{ y } ; \bx_\ell \, , \bx_\ell^* \, , \delta )$ and
estimate the marginal likelihood by
\begin{equation} \label{marginal_like_estim4}
\hspace*{-0.35in} \hat{ m }^{ PEP }_\ell ( \dn{ y } | \mt{ X }_\ell \, ,
\mt{ X }_\ell^* \, , \delta ) = m_0^N ( \dn{ y } | \bx_0, \bx_0^* ) \left[
\frac{ 1 }{ T } \sum_{ t = 1 }^T \frac{ m_\ell^N ( \dn{ y } | \, \dn{ y }^{
* ( t ) } ; \bx_\ell, \bx_\ell^* \, , \delta ) }{ m_0^N ( \dn{ y } | \,
\dn{ y }^{ * ( t ) } ; \bx_0, \bx_0^* \, , \delta ) } \frac{ m_0^N ( \dn{ y
}^{ * ( t ) } | \, \dn{ y } ; \bx_0, \bx_0^* \, , \delta ) }{ m_\ell^N (
\dn{ y }^{ * ( t ) } | \, \dn{ y } ; \bx_\ell, \bx_\ell^* \, , \delta ) }
\right] \, .
\end{equation}

\end{enumerate}

Monte-Carlo schemes (1) and (2) generate imaginary data from the posterior
predictive distribution of the model under consideration, and thus we
expect them to be relatively accurate. Moreover, in the second Monte-Carlo
scheme, when we estimate Bayes factors we only need to evaluate posterior
predictive distributions, which are available even in the case of improper
baseline priors. Closed-form expressions for the posterior predictive
distributions can be found in Section 2 of the web Appendix.

Using arguments similar to those in Section \ref{connection_j_z}, it is
clear that the marginal likelihoods $m^{ PEP }_\ell ( \dn{ y } | \mt{ X
}_\ell \, , \mt{ X }_\ell^* \, , \delta)$ under the two baseline prior
choices considered in this paper will yield the same posterior odds and
model probabilities for $g \rightarrow \infty$, $a_\ell = - \frac{ d_\ell
}{ 2 }$ and $b_\ell = 0$. This \vspace*{0.05in} is because the posterior
predictive densities involved in the expressions for $m^{ PEP }_\ell ( \dn{
y } | \mt{ X }_\ell \, , \mt{ X }_\ell^* \, , \delta)$ become the same for
the above-mentioned prior parameter values, while the corresponding prior
predictive density will be the same up to normalizing constants (common to
all models) that cancel out in the calculation.

\section{Consistency of the J-PEP Bayes factor}

\label{Limit}

Here we present a condensed version of a proof that Bayes factors based on
the J-PEP approach are consistent for model selection; additional details
are available in \cite{fouskakis-ntzoufras-2013}.

The PEP prior (\ref{pep}) can be rewritten as
\begin{eqnarray} \label{new-section-4-1}
\pi_\ell^{ PEP } ( \bbl, \varl | \bxl^*, \delta ) & = & \int \! \! \int
\pi_\ell^{ PEP } ( \bbl, \varl | \bbo, \varo ; \bxl^*, \delta ) \, \pi_0^N (
\bbo, \varo|\bxo^* ) \, d \bbo \, d \varo \, ,
\end{eqnarray}
in which the conditional PEP prior is given by
\begin{equation} \label{new-section-4-2}
\pi_\ell^{ PEP } ( \bbl, \varl | \bbo, \varo ; \bxl^*, \delta ) = \int
\frac{ f ( \by^* | \bbl, \varl, M_\ell ; \bxl^*, \delta ) \, f ( \by^* |
\bbo, \varo, M_0 ; \bx_0^*, \delta ) \, \pi^N ( \bbl, \varl | \bxl^* ) }{
m_\ell^N ( \by^* | \bxl^*, \delta ) } \, d \by^* \, .
\end{equation}
For the J-PEP prior, resulting from the baseline prior (\ref{JE}), it can
be shown --- following a line of reasoning similar to that in
\cite{moreno_etal_2003} --- that
\begin{eqnarray} \label{finalcondprior}
\pi_\ell^{ \textit{J-PEP} } ( \bbl, \varl | \bbo, \varo ; \bxl^*, \delta )
& = & \frac{ \Gamma ( \nstar - d_\ell ) }{ \Gamma ( \frac{ \nstar - d_\ell
}{ 2 } )^2 } ( \varo )^{ - \frac{ \nstar - d_\ell }{ 2 } } ( \varl )^{
\frac{ \nstar - d_\ell }{ 2 } - 1 } \left( 1 + \frac{ \varl }{ \varo }
\right)^{ - ( \nstar - d_\ell) } \, \times \nonumber \\
& & \hspace*{0.25in} f_{ N_{ \nstar } } \! \left[ \bbl ; \bboup, \delta (
\varl + \varo ) \invxtxstarl \right] \, ;
\end{eqnarray}
here $\bboup = ( \bbo^T, \dn{ 0 }_{ d_\ell - d_0 }^T )^T$ and $\dn{ 0 }_{ k
}$ is a vector of zeros of length $k$.

Following steps similar to those in \cite{moreno_etal_2003}, we find that
the Bayes factor of model $M_\ell$ versus the reference model $M_0$ (with
$M_0$ nested in $M_\ell$) is given by
\begin{equation} \label{bf_final}
BF_{ \ell \, 0 }^{ \textit{J-PEP} } = 2 \, \frac{ \Gamma \left( n - d_\ell
\right) }{ \Gamma \left( \frac{ n - d_\ell }{ 2 } \right)^2 } \bigints_{ 0
}^{ \frac{ \pi }{ 2 } } \frac{ ( \sin \phi )^{ n - d_0 - 1 } ( \cos \phi
)^{ n - d_\ell -1 } ( n + \sin^2 \phi )^{ \frac{ n - d_\ell }{ 2 } } }{
\left( n \frac{ RSS_\ell }{ RSS_0 } + \sin^2 \phi \right)^{ \frac{ n - d_0
}{ 2 } } } \, d \phi \, .
\end{equation}

\begin{theorem} \label{theorem 1}

For any two models $M_\ell$, $M_k \in {\cal M} \setminus \{ M_0 \}$ and for
large $n$, we have that
\begin{equation} \label{new-section-4-6}
-2 \, \log BF_{ \ell \, k }^{ \textit{J-PEP} } \approx n \log \frac{
RSS_\ell }{ RSS_k } + ( d_\ell - d_k ) \log n = BIC_\ell - BIC_k \, .
\end{equation}

\end{theorem}

\begin{proof}

For large $n$ we have that
\begin{eqnarray} \label{new-section-4-3}
( n + \sin^2 \phi )^{ \frac{ n - d_\ell }{ 2 } } & \approx & n^{ \frac{ n -
d_\ell }{ 2 } } \exp \left( \frac{ \sin^2 \phi }{ 2 } \right) \, , \nonumber
\\
\left( n \frac{ RSS_\ell }{ RSS_0 } + \sin^2 \phi \right)^{ \frac{ n - d_0
}{ 2 } } & \approx & \left( n \frac{ RSS_\ell }{ RSS_0 } \right)^{ \frac{ n
- d_0 }{ 2 } } \exp \left( \frac{ 1 }{ 2 } \sin^2 \phi \frac{ RSS_0 }{
RSS_\ell } \right) \, , \ \ \textrm{and} \\
\log \Gamma ( n - d_\ell ) - 2 \log \Gamma \left( \frac{ n - d_\ell }{ 2 }
\right) & \approx & \frac{ 1 }{ 2 } \log n + n \log 2 \, . \nonumber
\end{eqnarray}
From the above we obtain (\ref{new-section-4-6}), because of the integral 
inequality
\begin{equation} \label{new-section-4-5}
\bigint_{ 0 }^{ \frac{ \pi }{ 2 } } \frac{ ( \sin \phi )^{ n - d_0 - 1 } (
\cos \phi )^{ n - d_\ell - 1 } \exp \left( \frac{ \sin^2 \phi }{ 2 } \right)
}{ \exp \left( \frac{ 1 }{ 2 } \sin^2 \phi \, \frac{ RSS_0 }{ RSS_\ell }
\right) } \, d \phi \le \int_{ 0 }^{ \frac{ \pi }{ 2 } } \exp \left[ \frac{
\sin^2 \phi }{ 2 } \left( 1 - \frac{ RSS_0 }{ RSS_\ell } \right) \right] \,
d \phi \, ,
\end{equation}
which is true for any $n \ge ( d_0 + 1 )$ and $n \ge ( d_\ell + 1 )$.
\citet[p.~1216]{casella_etal_2009} have shown that the right-hand integral
in (\ref{new-section-4-5}) is finite for all $n$; therefore the left-hand
integral in (\ref{new-section-4-5}), which arises in the computation of
$BF_{ \ell \, 0 }^{ \textit{J-PEP} }$ via equation (\ref{bf_final}), is
also finite for all $n$.
\end{proof}

Therefore the J-PEP approach has the same asymptotic behavior as the
BIC-based variable-selection procedure. The following Lemma is a direct
result of (a) Theorem \ref{theorem 1} above and (b) Theorem 4 of
\cite{casella_etal_2009}.

\vspace*{0.1in}

\begin{lemma} \label{lemma 1}

Let $M_\ell \in {\cal M}$  be a Gaussian regression model of type
(\ref{new2-1}) such that
\[
\lim \limits_{ n \rightarrow \infty } \frac{ \mt{ X }_T \left[ \mt{ I }_n -
\mt{ X }_\ell ( \mt{ X }_\ell^T \, \mt{ X }_\ell )^{ -1 } \mt{ X }_\ell^T
\right] \mt{ X }_T }{ n } \mbox{ is a positive semi-definite matrix},
\]
in which $X_T$ is the design matrix of the true data-generating regression
model $M_T \neq M_j$. \vspace*{0.05in} Then the variable selection
procedure based on the J-PEP Bayes factor is consistent, since $BF^{ J-PEP
}_{ jT } \rightarrow 0$ as $n \rightarrow \infty$.

\end{lemma}

\section{Experimental results}

\label{examples}

In this Section we illustrate the PEP-prior methodology with two case
studies --- one simulated, one real --- and we perform sensitivity analyses
to verify the stability of our findings; results are presented for both
Z-PEP and J-PEP. In both cases, the marginal likelihood (\ref{new2-17}) is
not analytically tractable, and therefore initially we evaluate the four
Monte-Carlo marginal-likelihood approaches given in Section
\ref{marginal_likelihood_pep_gen} above and in Section 3 of the web
Appendix. Then we present results for $n^* = n$, followed by an extensive
sensitivity analysis over different values of $n^*$. Our results are
compared with those obtained using (a) the EPP with minimal training sample,
power parameter $\delta = 1$ and the independence Jeffreys prior as baseline
(we call this approach J-EPP) and (b) the expected intrinsic Bayes factor
(EIBF), i.e., the arithmetic mean of the IBFs over different minimal
training samples (in Section \ref{comparisons-ibf-j-epp} we also make some
comparisons between the Z-PEP, J-PEP and IBF methods). Implementation
details for J-EPP can be found in \cite{fouskakis_ntzoufras_2012}, while
computational details for the EIBF approach are provided in Section 4 of the
web Appendix. In all illustrations the design matrix $\mt{ X }^*$ of the
imaginary/training data is selected as a random subsample of size $n^*$ of
the rows of $\mt{ X }$.

Note that, since \cite{perez_berger_2002} have shown that Bayes factors
from the J-EPP approach become identical to those from the EIBF method as
the sample size $n \rightarrow \infty$ (with the number of covariates $p$
fixed), it is possible (for large $n$) to use EIBF as an approximation to
J-EPP that is computationally much faster than the full J-EPP calculation.
We take advantage of this fact below: for example, producing the results in
Table \ref{ozonetab42} would have taken many days of CPU time with J-EPP;
instead, essentially equivalent results were available in hours with EIBF.
For this reason, one can regard the labels ``J-EPP" and ``EIBF" as more or
less interchangeable in what follows.

\subsection{A simulated example}

\label{sim}

Here we illustrate the PEP method by considering, as a case study, the
simulated data set of \cite{nott_kohn_05}. This data set consists of $n =
50$ observations with $p = 15$ covariates. The first 10 covariates are
generated from a multivariate Normal distribution with mean vector $\dn{ 0
}$ and covariance matrix $\mt{ I }_{ 10 }$, while
\begin{equation} \label{new3-1}
X_{ ij } \sim N \big( 0.3 X_{ i1 } + 0.5 X_{ i2 } + 0.7 X_{ i3 } + 0.9 X_{
i4 } + 1.1 X_{ i5 }, 1 \big) \mbox{ for } ( j = 11, \dots, 15; i = 1,
\dots, 50 ) \, ,
\end{equation}
and the response is generated from
\begin{equation} \label{ss}
Y_i \sim N \big( 4 + 2 X_{ i1 } - X_{ i5 } + 1.5 X_{ i7 } + X_{ i, 11} +
0.5 X_{ i, 13 }, 2.5^2 \big) \mbox{ for } i = 1, \dots, 50 \, .
\end{equation}
With $p = 15$ covariates there are only 32,768 models to compare; we were
able to conduct a full enumeration of the model space, obviating the need
for a model-search algorithm in this example.

\subsubsection{PEP prior results}

\label{pi-g-prior-1}

To check the efficiency of the four Monte-Carlo marginal-likelihood
estimates (the first two of which are detailed in Section
\ref{marginal_likelihood_pep_gen} above, and the second two in Section 3 of
the web Appendix), we initially performed a small experiment. For Z-PEP, we
estimated the logarithm of the marginal likelihood for models $( X_1 + X_5 +
X_7 + X_{ 11 } )$ and $( X_1 + X_7 +X_{ 11 } )$, by running each Monte-Carlo
technique 100 times for 1,000 iterations and calculating the Monte-Carlo
standard errors. For both models the first and second Monte-Carlo schemes
produced Monte-Carlo standard errors of approximately 0.03, while the
Monte-Carlo standard errors of the third and fourth schemes were larger by
multiplicative factors of 30 and 20, respectively. In what follows,
therefore, we used the first and second schemes; in particular we employed
the first scheme for Z-PEP and the second scheme for J-PEP, holding the
number of iterations constant at 1,000.

\begin{table}[t!]

\caption{\textit{Posterior model probabilities for the best models,
together with Bayes factors for the Z-PEP MAP model ($M_1$) against $M_j, j
= 2, \dots, 7$, for the Z-PEP and J-PEP prior methodologies, in the
simulated example of Section \ref{sim}.}}

\label{ex2tab1b}

\begin{center}

\begin{tabular}{c|l|cc|ccc}

\multicolumn{1}{c}{} & & \multicolumn{2}{c|}{Z-PEP} &
\multicolumn{3}{c}{J-PEP} \\ \cline {3-7}

\multicolumn{1}{c}{} & & Posterior Model & Bayes & & Posterior Model &
Bayes\\

$M_j $ & \multicolumn{1}{c|}{Predictors} & Probability & Factor & Rank &
Probability & Factor \\

\hline

1 & $X_1 + X_5 + X_7 + X_{ 11 }$ & 0.0783 & 1.00 & (2) & 0.0952 & 1.00 \\

2 & $X_1 + X_7 + X_{ 11 }$ & 0.0636 & 1.23 & (1) & 0.1054 & 0.90 \\

3 & $X_1 + X_5 + X_6 + X_7 + X_{ 11 }$ & 0.0595 & 1.32 & (3) & 0.0505 &
1.88 \\

4 & $X_1 + X_6 + X_7 + X_{ 11 }$ & 0.0242 & 3.23 & (4) & 0.0308 & 3.09 \\

5 & $X_1 + X_7 + X_{ 10 } + X_{ 11 }$ & 0.0175 & 4.46 & (5) & 0.0227 & 4.19
\\

6 & $X_1 + X_5 + X_7 + X_{ 10 } + X_{ 11 }$ & 0.0170 & 4.60 & (9) & 0.0146
& 6.53 \\

7 & $X_1 + X_5 + X_7 + X_{ 11 } + X_{ 13 }$ & 0.0163 & 4.78 & (10) & 0.0139
& 6.87 \\

\end{tabular}

\end{center}

\vspace*{-0.25in}

\end{table}

Table \ref{ex2tab1b} presents the posterior model probabilities (with a
uniform prior on the model space) for the best models in (a single
realization of) the Nott-Kohn model, together with Bayes factors, for the
Z-PEP and J-PEP prior methodologies. The maximum a-posteriori (MAP) model
for the Z-PEP prior includes four of the five true effects; the
data-generating model is seventh in rank due to the small effect of $X_{ 13
}$. Moreover, note that when using the J-PEP prior the methodology is more
parsimonious; the MAP model is now $X_1 + X_7 + X_{ 11 }$, which is the
second-best model under the Z-PEP approach. When we focus on posterior
inclusion probabilities (results omitted for brevity) rather than posterior
model probabilities and odds, J-PEP supports systematically more
parsimonious models than Z-PEP, but no noticeable differences between the
inclusion probabilities using the two priors are observed (with the largest
difference seen in the inclusion probabilities of $X_5$; these are about 0.5
for Z-PEP and about 0.4 for J-PEP).

\subsubsection{Sensitivity analysis for the imaginary/training sample size
$n^*$}

\label{sensitivity-analysis-for-n-star}

To examine the sensitivity of the PEP approach to the sample size $n^*$ of
the imaginary/training data set, we present results for $n^* = 17, \dots,
50$: Figure \ref{ex2plot3b} displays posterior marginal variable-inclusion
probabilities (in the same single realization of the Nott-Kohn model that
led to Table \ref{ex2tab1b}). As noted previously, to specify $\mt{ X }^*$
when $n^* < n$ we randomly selected a subsample of the rows of the original
matrix $\mt{ X }$. Results are presented for Z-PEP; similar results for
J-PEP have been omitted for brevity. It is evident that posterior inclusion
probabilities are quite insensitive to a wide variety of values of $n^*$,
while more variability is observed for smaller values of $n^*$; this arises
from the selection of the subsamples used for the construction of $\mt{ X
}^*$. The picture for the posterior model probabilities (not shown) is
similar.

\begin{figure}[t!]

\centering

\caption{\textit{Posterior marginal inclusion probabilities, for $n^*$
values from $17$ to $n = 50$, with the Z-PEP prior methodology, in the
simulated example of Section \ref{sim}.}}

\vspace*{-0.2in}

\psfrag{nstar}[c][c][0.9]{Imaginary/Training Data Sample Size ($n^*$)}

\includegraphics[ scale = 0.7 ]{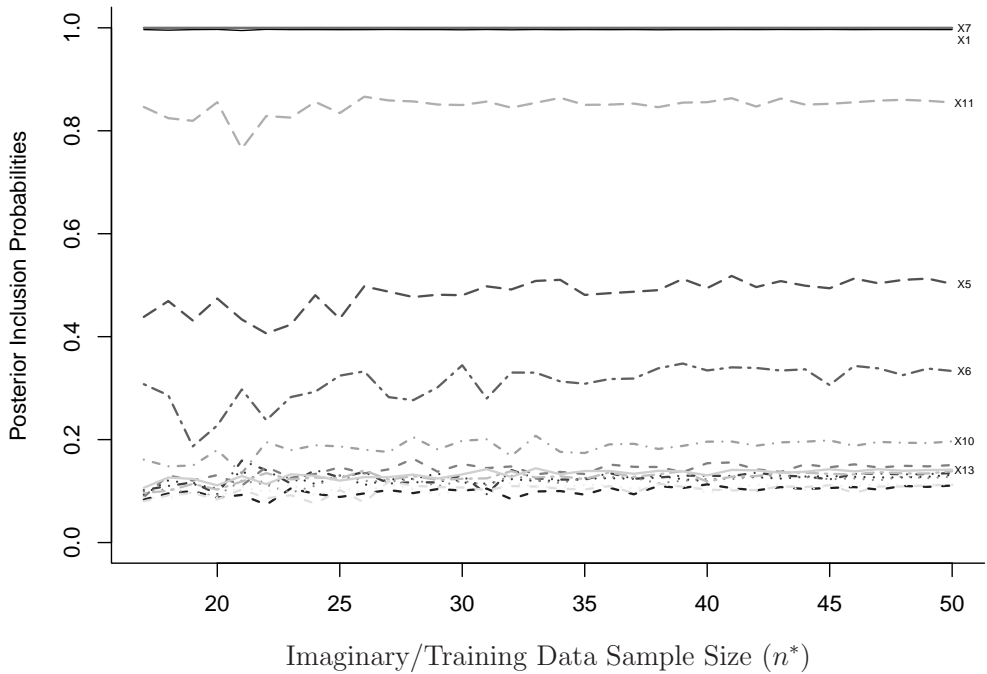}

\vspace*{-0.15in}

\label{ex2plot3b}

\end{figure}

To further examine the stability of this conclusion, we generated an
additional 50 data sets from the Nott-Kohn sampling scheme (\ref{new3-1},
\ref{ss}) and repeated the analysis that led to Figure \ref{ex2plot3b}, in
this case for all of the true non-zero effects in this model ($X_1, X_5,
X_7, X_{ 11 }$ and $X_{ 13 }$). The evolution of the posterior marginal
inclusion probabilities as a function of $n^*$ for each of the non-zero
effects is presented in the right-hand column of Figure
\ref{sens_nstar_sim_st}; in the left-hand column the corresponding medians
and quartiles of the same quantities (over all 50 samples) are depicted. The
results are similar to those in Figure \ref{ex2plot3b}: for each data set,
posterior marginal inclusion probabilities are remarkably insensitive to a
wide variety of values of $n^*$. We draw the key conclusion from these
analyses that one can use $n^* = n$ and dispense with training samples
altogether in the PEP methodology; this yields all the advantages mentioned
earlier (increased stability of the resulting Bayes factors, removal of the
arbitrariness arising from individual training-sample selections, and
substantial increases in computational speed, allowing many more models to
be compared within a fixed CPU budget).

\begin{figure}[p]

\centering

\caption{\textit{Posterior marginal inclusion probabilities $P( \gamma_j = 1
| \dn{ y } )$ (right column) together with their medians and quartiles (left
column) over the 50 additional samples from the Nott-Kohn model, for each of
the non-zero effects ($j \in \{ 1, 5, 7, 11, 13 \}$) and for $n^*$ varying
from 17 to 50.}}

\label{sens_nstar_sim_st}

\vspace*{0.2in}

\psfrag{n*}[c][c][0.9]{$n^*$}

\psfrag{PIPX1}[c][c][0.8]{$P(\gamma_1=1 | \dn{y})$}

\psfrag{PIPX5}[c][c][0.8]{$P(\gamma_5=1 | \dn{y})$}

\psfrag{PIPX7}[c][c][0.8]{$P(\gamma_7=1 | \dn{y})$}

\psfrag{PIPX11}[c][c][0.8]{$P(\gamma_{11}=1 | \dn{y})$}

\psfrag{PIPX13}[c][c][0.8]{$P(\gamma_{13}=1 | \dn{y})$}

\includegraphics[ scale = 0.725 ]{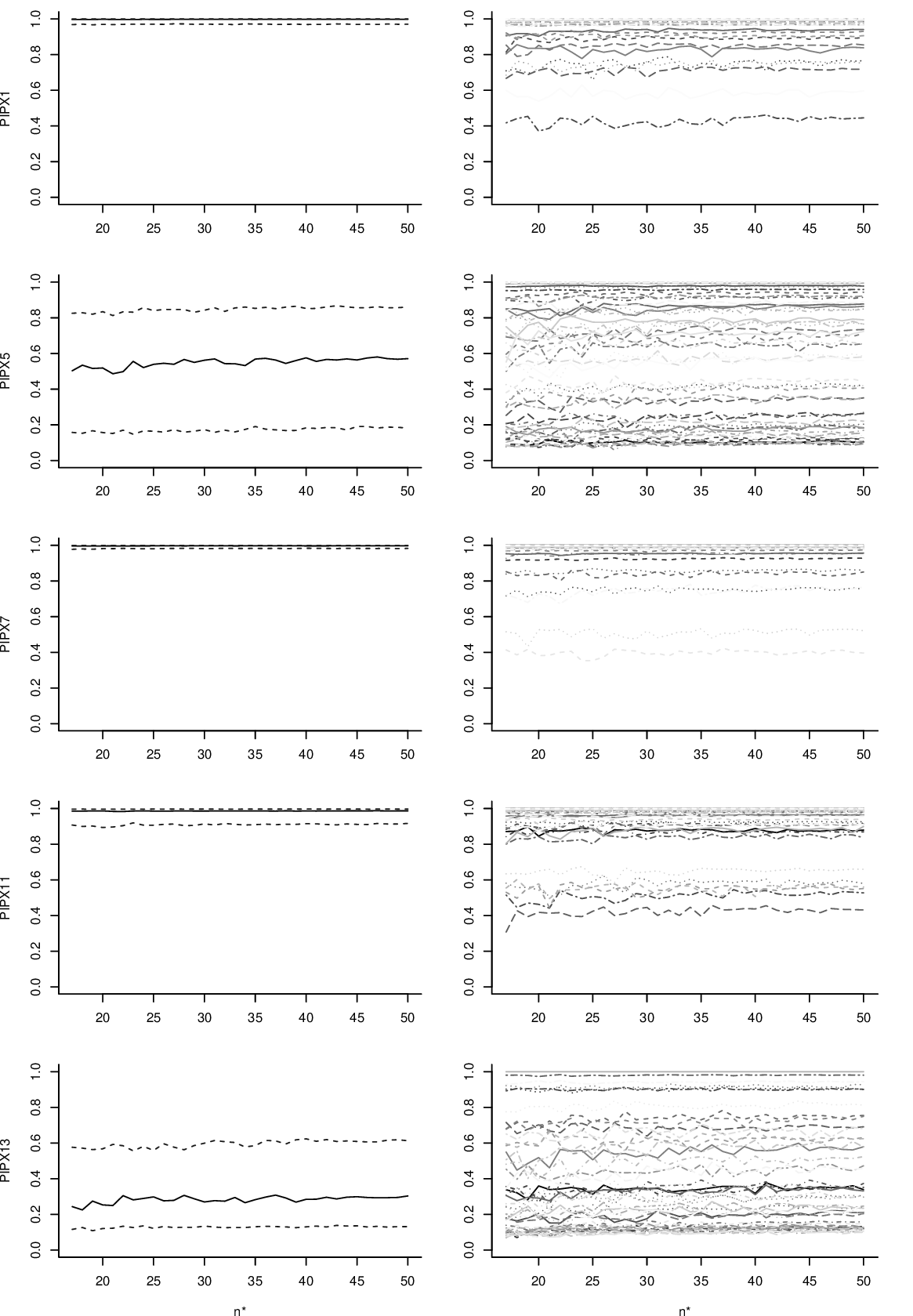}

\end{figure}

One of the main features of PEP is its unit-information property, an
especially important consideration when $p$ is a substantial fraction of
$n$; as noted in Section \ref{introduction}, this situation arises with some
frequency in disciplines such as economics and genomics. In contrast to PEP,
the EPP --- which is equivalent to the intrinsic prior --- can be highly
influential when $n$ is not much larger than $p$. To illustrate this point
we kept the first $n = 20$ observations from the single simulated data set
that led to Figure \ref{ex2plot3b} and considered a randomly selected
training sample of minimal size ($n^* = 17$). Figure
\ref{posterior_box_plots} presents the posterior distribution of the
regression coefficients for PEP ($\delta = n^*$) and for EPP ($\delta = 1$),
in comparison with the MLEs (solid horizontal lines). From this figure it is
clear that the PEP prior produces posterior results identical to the MLEs,
while EPP has a substantial unintended impact on the posterior distribution
(consider in particular the marginal posteriors for $\beta_2, \beta_7,
\beta_9, \beta_{ 11 }$ and $\beta_{ 12 }$). Moreover, the variability of the
resulting posterior distributions using the PEP approach is considerably
smaller (in this regard, consider especially the marginal posteriors for
$\beta_5, \beta_7$ and $\beta_{ 11 }$).

\begin{figure}[t!]

\centering

\caption{\textit{Boxplots of the posterior distributions of the regression
coefficients. For each coefficient, the left-hand boxplot summarizes the EPP
results and the right-hand boxplot displays the Z-PEP posteriors; solid
lines in both posteriors identify the MLEs. We used the first 20
observations from the simulated data-set, in the example of Section
\ref{pi-g-prior-1} that led to Figure \ref{ex2plot3b}, and a randomly
selected training sample of size $n^* = 17$.}}

\label{posterior_box_plots}

\vspace*{-0.3in}

\includegraphics[ scale = 0.8 ]{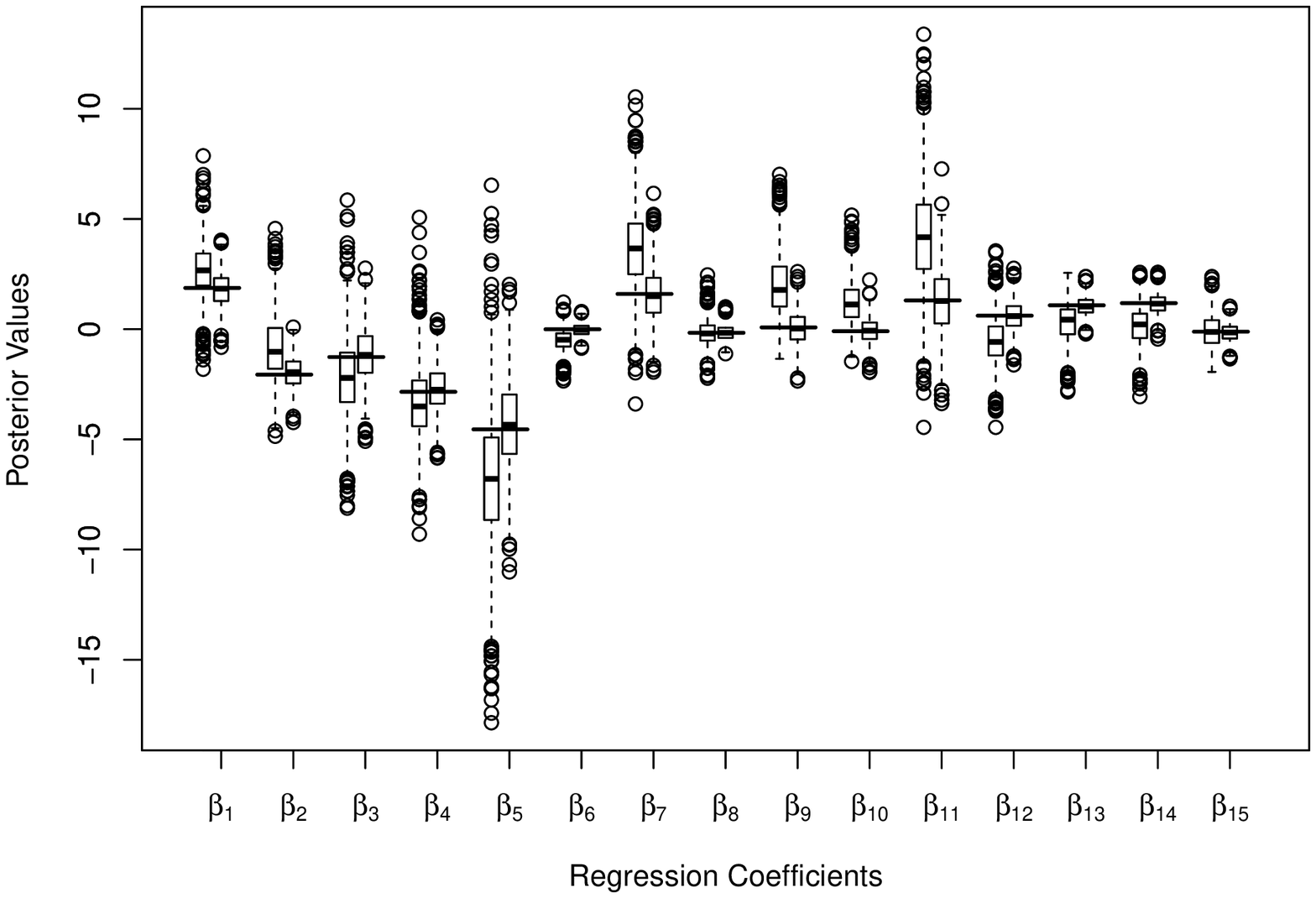}

\vspace*{-0.1in}

\end{figure}

\subsubsection{Comparisons with the intrinsic-Bayes-factor (IBF) and J-EPP
approaches}

\label{comparisons-ibf-j-epp}

Here we compare the PEP Bayes factor between the two best models ($( X_1 +
X_5 + X_7 + X_{ 11 } )$ and $( X_1 + X_7 + X_{ 11 } )$) with the
corresponding Bayes factors using J-EPP and IBF. For IBF and J-EPP we
randomly selected 100 training samples of size $n^* = 6$ (the minimal
training sample size for the estimation of these two models) and $n^* = 17$
(the minimal training sample size for the estimation of the full model with
all $p = 15$ covariates), while for Z-PEP and J-PEP we randomly selected
100 training samples of sizes $n^* = \{ 6, 17, 30, 40, 50 \}$. Each
marginal-likelihood estimate in PEP was obtained with 1,000 iterations,
using the first and second Monte-Carlo schemes for Z-PEP and J-PEP,
respectively, and in J-EPP with 1,000 iterations, using the second
Monte-Carlo scheme. Figure \ref{ex2comparison1a} presents the results as
parallel boxplots, and motivates the following observations:

\begin{figure}[t!]

\centering

\caption{\textit{Boxplots of the Intrinsic Bayes Factor (IBF) and Bayes
factors using the J-EPP, Z-PEP and J-PEP approaches, on a logarithmic scale,
in favor of model $( X_1 + X_5 + X_7 + X_{ 11 } )$ over model $( X_1 + X_7 +
X_{ 11 } )$ in the simulated example of Section \ref{pi-g-prior-1}. For IBF
and J-EPP, training samples of size $n^* =$ 6 and 17 were used; for both PEP
priors we used $n^* = \{ 6, 17, 30, 40, 50 \}$. In the boxplot labels on the
vertical axis, letters indicate methods and numbers signify training sample
sizes.}}

\label{ex2comparison1a}

\vspace*{-0.25in}

\psfrag{PEPP_J50}[c][c][0.9]{\footnotesize \sf J-PEP 50 \hspace{1em}}

\psfrag{PEPP_J40}[c][c][0.9]{\footnotesize \sf J-PEP 40 \hspace{1em}}

\psfrag{PEPP_J30}[c][c][0.9]{\footnotesize \sf J-PEP 30 \hspace{1em}}

\psfrag{PEPP_J17}[c][c][0.9]{\footnotesize \sf J-PEP 17 \hspace{1em}}

\psfrag{PEPP_J6}[c][c][0.9]{\footnotesize \sf J-PEP 6 \hspace{1em}}

\psfrag{PEPP_Z50}[c][c][0.9]{\footnotesize \sf Z-PEP 50 \hspace{1em}}

\psfrag{PEPP_Z40}[c][c][0.9]{\footnotesize \sf Z-PEP 40 \hspace{1em}}

\psfrag{PEPP_Z30}[c][c][0.9]{\footnotesize \sf Z-PEP 30 \hspace{1em}}

\psfrag{PEPP_Z17}[c][c][0.9]{\footnotesize \sf Z-PEP 17 \hspace{1em}}

\psfrag{PEPP_Z6}[c][c][0.9]{\footnotesize \sf Z-PEP 6 \hspace{1em}}

\psfrag{EPP_J17}[c][c][0.9]{\footnotesize \sf J-EPP 17 \hspace{1em}}

\psfrag{EPP_J6}[c][c][0.9]{\footnotesize \sf J-EPP 6 \hspace{1em}}

\psfrag{EPP_J17}[c][c][0.9]{\footnotesize \sf J-EPP 17 \hspace{1em}}

\psfrag{EPP_J6}[c][c][0.9]{\footnotesize \sf J-EPP 6 \hspace{1em}}

\psfrag{IBF_17}[c][c][0.9]{\footnotesize \sf IBF 17 \hspace{1em}}

\psfrag{IBF_6}[c][c][0.9]{\footnotesize \sf IBF 6 \hspace{1em}}

\psfrag{Log Bayes Factor}[c][c][1.0]{\sf Log Bayes Factor}

\includegraphics[ scale = 0.6, angle = -90 ]{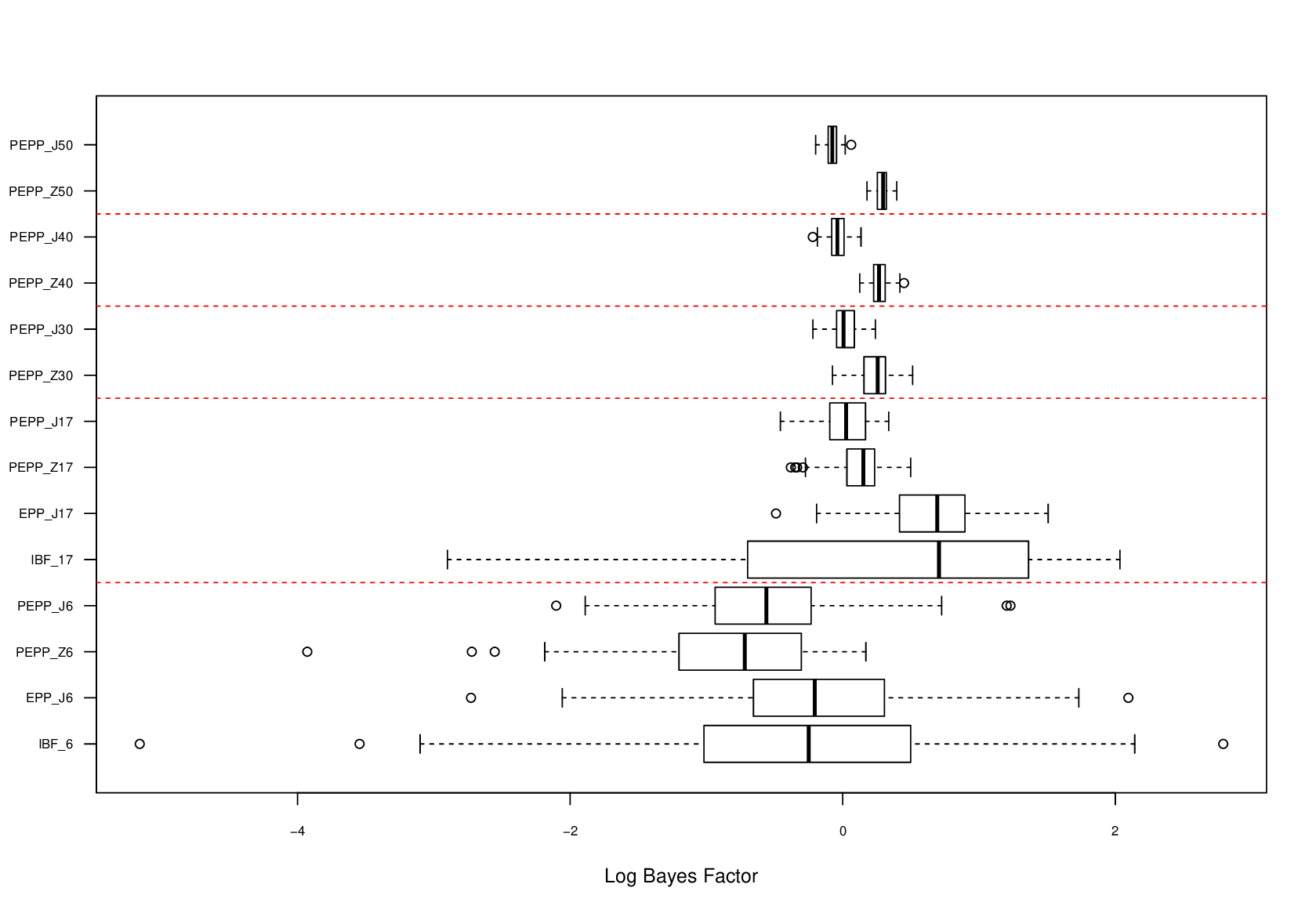}

\end{figure}

\begin{itemize}

\item

For $n^* =$ 6 and 17, although there are some differences between the
median log Bayes factors across the four approaches, the variability across
random training samples is so large as to make these differences small by
comparison; none of the methods finds a marked difference between the two
models.

\item

With modest $n^*$ values, which would tend to be favored by users for their
advantage in computing speed, the IBF method exhibited an extraordinary
amount of instability across the particular random training samples chosen:
with $n^* = 6$ the observed variability of IBF estimated Bayes factors
across the 100 samples was from $e^{ -5.16 } \doteq 0.005$ to $e^{ +2.48 }
\doteq 11.89$, a multiplicative range of more than 2,300, and with $n^* =
17$ the corresponding span was from $e^{ -2.90 } \doteq 0.055$ to $e^{
+2.03 } \doteq 7.61$, a multiplicative variation of about 138. (This
instability was observed by the original authors of IBF
\citep{berger_pericchi_96b}, and for this reason they recommended the use
of either the Median IBF or theoretical intrinsic priors. These
recommendations were combined with the Cauchy-Binet Theorem in order to
compute an average of determinants of sub-matrices required for these
quantities; see, e.g., \cite{berger_pericchi_2004}.)

The instability of the J-EPP approach across training samples was smaller
than with IBF but still large: for J-EPP the range of estimated Bayes
factors for $n^* = 6$ was from $e^{ -2.72 } \doteq 0.065$ to $e^{ +2.09 }
\doteq 8.08$ (a multiplicative span of about 125); the corresponding values
for $n^* = 17$ were from 0.61 to 4.51, a multiplicative range of 7.4. The
analogous multiplicative spans for Z-PEP were considerably smaller: 60.22,
2.41 and 1.24, respectively, for $n^* = 6, 17$ and 50; similarly for J-PEP
the corresponding multiplicative ranges were 28.01, 2.21 and 1.30.

\item

Figure \ref{ex2comparison1a} highlights the advantage of using $n^* = n$
with the PEP approach over the IBF and J-EPP methods with modest training
samples: the Monte-Carlo uncertainty introduced in the IBF and J-EPP
methods by the need to choose a random training sample creates a remarkable
degree of sensitivity in those approaches to the particular samples chosen,
and this undesirable behavior is entirely absent with the $n^* = n$ version
of the PEP method. The observed variability for $n^* = n$ in the PEP
approach is due solely to Monte-Carlo noise in the marginal-likelihood
computation.

\end{itemize}

\subsection{Variable selection in the Breiman-Friedman ozone data set}

\label{ozone_data}

In this Section we use, as a second case study, a data set often examined in
variable-selection studies --- the ozone data of
\cite{breiman_friedman_1985} --- to implement the Z-PEP and J-PEP approaches
and make comparisons with other methods. The scientific purpose of building
this data set was to study the relationship between ozone concentration and
a number of meteorological variables, including temperature, wind speed,
humidity and atmospheric pressure; the data are from a variety of locations
in the Los Angeles basin in 1976. The data set we used was slightly modified
from its form in other studies, based on preliminary exploratory analyses we
performed; our version of the data set has $n = 330$. As a response we used
a standardized version of the logarithm of the ozone variable of the
original data set. The standardized versions of 9 main effects, 9 quadratic
terms, 2 cubic terms, and 36 two-way interactions (a total of 56 explanatory
variables) were included as possible covariates. (Further details concerning
the final data set used in this Section are provided in Section 5 of the web
Appendix.)

\subsubsection{Searching the model space}

\label{model-space-search}

Full-enumeration search for the full space with 56 covariates was
computationally infeasible, so we used a model-search algorithm (based on
$MC^3$), given in Section 6 of the web Appendix, for the Z-PEP prior
methodology and the EIBF approach. For Z-PEP we used the first Monte-Carlo
marginal-likelihood scheme with 1,000 iterations; for EIBF we employed 30
randomly-selected minimal training samples ($n^* = 58$).

With such a large number of predictors, the model space in our problem was
too large for the $MC^3$ approach to estimate posterior model probabilities
with high accuracy in a reasonable amount of CPU time. For this reason, we
implemented the following two-step method:

\begin{enumerate}

\item[(1)]

First we used $MC^3$ to identify variables with high posterior marginal
inclusion probabilities $P ( \gamma_j = 1 | \by )$, and we then created a
reduced model space consisting only of those variables whose marginal
probabilities were above a threshold value. According to
\cite{barbieri_berger_2004}, this method of selecting variables may lead to
the identification of models with better predictive abilities than
approaches based on maximizing posterior model probabilities. Although
Barbieri and Berger proposed 0.5 as a threshold value for $P ( \gamma_j = 1
| \by )$, we used the lower value of 0.3, since our aim was only to
identify and eliminate variables not contributing to models with high
posterior probabilities. The inclusion probabilities were based on the
marginal-likelihood weights for the visited models.

\item[(2)]

Then we used the same model search algorithm as in step (1) in the reduced
space to estimate posterior model probabilities (and the corresponding
odds).

\end{enumerate}

Initially we ran $MC^3$ for 100,000 iterations for both the Z-PEP and EIBF
approaches. The reduced model space was formed from those variables that
had posterior marginal inclusion probabilities above 0.3 in either run.
With this approach we reduced the initial list of $p = 56$ available
candidates down to 22 predictors; Section 7 in the web Appendix lists these
covariates.

In the reduced model space we then ran $MC^3$ for 220,000 iterations for
the J-PEP, Z-PEP and EIBF approaches. For J-PEP we used the second
Monte-Carlo scheme with 1,000 iterations, for Z-PEP we employed the first
Monte-Carlo scheme (also with 1,000 iterations), and for EIBF we used 30
randomly-selected minimal training samples ($n^* = 24$). The resulting
posterior model odds for the five best models under each approach are given
in Table \ref{ozonetab42}. The MAP model under the Z-PEP approach was the
only one that appeared in the five most probable models in all approaches
(with rank 2 in J-PEP and rank 5 in EIBF). From this table it is clear that
the J-PEP approach supports the most parsimonious models; at the other
extreme, EIBF gives the least support to the most parsimonious models. When
attention is focused on posterior inclusion probabilities (not shown here),
the conclusions are similar: the three methods give approximately equal
support to the most prominent covariates, while for the less important
predictors the posterior inclusion probabilities are highest for EIBF,
lower for Z-PEP, and lowest for J-PEP. This confirms that the PEP
methodology supports more parsimonious models than the EIBF approach.

\begin{table}[t!]

\begin{center}

\caption{\textit{Posterior odds ($PO_{ 1 k }$) of the five best models
within each analysis versus the current model $k$, for the reduced model
space of the ozone data set. Variables common in all three analyses were
$X_1 + X_2  + X_8 + X_9 + X_{ 10 } + X_{ 15 } + X_{ 16 } + X_{ 18 } + X_{ 43
}$.}}

\label{ozonetab42}

\begin{tabular}{ccc|c@{}c@{}c@{}c@{ }|c@{ }|c}

\multicolumn{8}{c}{} \\

\multicolumn{8}{c}{J-PEP} \\

\hline

\multicolumn{3}{c|}{Ranking} & & & & & Number of & Posterior \\ \cline{1-3}

J-PEP & Z-PEP & EIBF & \multicolumn{4}{c|}{Additional Variables} &
Covariates & Odds $PO_{ 1 k }$ \\

\hline \hline

1 & ($>$5) & ($>$5) & & & & & 9 & 1.00 \\

2 & (1) & (5) & & $X_7 + X_{ 12 } + X_{ 13 }$ & $+ X_{ 20 }$ & & 13 & 1.29
\\

3 & ($>$5) & ($>$5) & & $X_7~~~~~~~~~ + X_{ 13 }$ & $+ X_{ 20 }$ & & 12 &
1.46 \\

4 & ($>$5) & ($>$5) & & $~~~~~~~~~X_{ 12 }~~~~~~~~~$ & $+ X_{ 20 }$ & & 11
& 1.87 \\

5 & ($>$5) & ($>$5) & & $~~~~~~~~~X_{ 12 }~~~~~~~~~$ &  & & 10 & 2.08 \\

\multicolumn{8}{c}{}\\

\multicolumn{8}{c}{Z-PEP}\\

\hline

\multicolumn{3}{c|}{Ranking} & & & & & Number of & Posterior \\ \cline{1-3}

Z-PEP & J-PEP & EIBF & \multicolumn{4}{c|}{Additional Variables} &
Covariates & Odds $PO_{ 1 k }$ \\

\hline \hline

1 & (2) &(5) & & $X_7 + X_{ 12 } + X_{ 13 }$ & $+ X_{ 20 }$ & & 13 & 1.00
\\

2 & ($>$5) & ($>$5) & $X_5 +$ & $X_7 + X_{ 12 } + X_{ 13 }$ & $+ X_{ 20 }$
& & 14 & 1.19 \\

3 & ($>$5) & (3) & $X_5 +$ & $X_7 + X_{ 12 } + X_{ 13 }$ & $+ X_{ 20 }$ &
$+ X_{ 42 }$ & 15 & 1.77 \\

4 & ($>$5) & (1) & & $X_7 + X_{ 12 } + X_{ 13 }$ & $+ X_{ 20 }$ & $+ X_{ 42
}$ & 14 & 1.94 \\

5 & ($>$5) & ($>$5) & & $X_7 + X_{ 12 } + X_{ 13 }$ & & & 12 & 2.30 \\

\multicolumn{8}{c}{} \\

\multicolumn{8}{c}{EIBF} \\

\hline

\multicolumn{3}{c|}{Ranking} & & & & & Number of & Posterior \\ \cline{1-3}

EIBF & J-PEP & Z-PEP & \multicolumn{4}{c|}{Additional Variables} &
Covariates & Odds $PO_{ 1 k }$ \\

\hline \hline

1 & ($>$5) &(4) & & $X_7 + X_{ 12 } + X_{ 13 }$ & $+ X_{ 20 }$ & $+ X_{ 42
}$ & 14 & 1.00 \\

2 & ($>$5) & ($>$5) & $X_5 +$ & $X_7 + X_{ 12 } + X_{ 13 }$ &$+ X_{ 20 }$ &
$+ X_{ 26 } + X_{ 42 }$ & 16 & 1.17 \\

3 & ($>$5) & (3) & $X_5 +$ & $X_7 + X_{ 12 } + X_{ 13 }$ &$+ X_{ 20 }$ & $+
X_{ 42 }$ & 15 & 1.30 \\

4 & ($>$5) & ($>$5) & & $X_7 + X_{ 12 } + X_{ 13 }$ &$+ X_{ 20 }$ & $+ X_{
39 } + X_{ 42 }$ & 15 & 1.44 \\

5 & (2) & (1) & & $X_7 + X_{ 12 } + X_{ 13 }$ & $+ X_{ 20 }$ & & 13 & 1.58
\\

\end{tabular}

\end{center}

\end{table}

\begin{table}[t!]

\caption{\textit{Comparison of the predictive performance of the PEP and
J-EPP methods, using the full and MAP models in the reduced model space of
the ozone data set.}}

\label{ozonetab5}

\vspace*{-0.1in}

\begin{center}

\begin{tabular}{c@{~}|@{~}c@{~}c@{~}c@{~}|c@{~~~}c@{~~~}c@{~~~}c@{~}}

& & & & \multicolumn{4}{c}{$RMSE^*$} \\ \cline{5-8}

Model & $d_\ell$ & $R^2$ & $R^2_{ adj }$ & J-PEP & Z-PEP & J-EPP & Jeffreys
Prior \\

\hline \hline

Full & 22 & 0.8500 & 0.8392 & 0.5988 & 0.5935 & 0.6194 & 0.5972 \\

& & & & (0.0087) & (0.0097) & (0.0169) & (0.0104) \\

\hline

J-PEP MAP & \ 9 & 0.8070 & 0.8016 & 0.5975 & 0.6161 & 0.7524 & 0.6165 \\

& & & & (0.0063) & (0.0051) & (0.0626) & (0.0052) \\

\hline

Z-PEP MAP & 13 & 0.8370 & 0.8303 & 0.5994 & 0.5999 & 0.6982 & 0.5994 \\

& & & & (0.0071) & (0.0060) & (0.0734) & (0.0049) \\

\hline

EIBF MAP & 14 & 0.8398 & 0.8326 & 0.6182 & 0.5961 &
0.6726 & 0.5958 \\

& & & & (0.0066) & (0.0072) & (0.0800) & (0.0061) \vspace*{0.1in} \\

\multicolumn{8}{c}{\underline{Comparison with the full model (percentage
changes)}} \vspace*{0.1in} \\

& & & & \multicolumn{4}{c}{$RMSE$} \\ \cline{5-8}

Model & $d_\ell$ & $R^2$ & $R^2_{ adj }$ & J-PEP & Z-PEP & J-EPP & Jeffreys
Prior \\

\hline

J-PEP MAP & $-59\%$ & $-5.06\%$ & $-4.48\%$ & $-0.22\%$ & $+3.81\%$ &
$+21.5\%$ & $+3.23\%$ \\

Z-PEP MAP & $-41\%$ & $-1.50\%$ & $-1.06\%$ & $+0.10\%$ & $+1.01\%$ &
$+12.7\%$ & $+0.37\%$ \\

EIBF MAP & $-36\%$ & $-1.20\%$ & $-0.78\%$ & $+3.24\%$ & $+0.44\%$ &
$+10.9\%$ & $-0.23\%$ \\

\end{tabular}

\end{center}

\footnotesize

\textit{Note:} $^*$Mean (standard deviation) over 50 different split-half
out-of-sample evaluations.

\normalsize

\end{table}

\begin{figure}[t!]

\centering

\caption{\textit{Distribution of $RMSE$ across 50 random partitions of the
ozone data set, for the Jeffreys-prior, J-EPP, Z-PEP and J-PEP methods, in
(a) the full model, (b) the Z-PEP MAP model, (c) the J-EPP MAP model, and
(d) the J-PEP MAP model.}}

\label{armse1}

\vspace*{0.1in}

\psfrag{PEPP_J}[c][c][0.9]{\footnotesize \sf J-PEP \hspace{1em}}

\psfrag{PEPP_Z}[c][c][0.9]{\footnotesize \sf Z-PEP \hspace{1em}}

\psfrag{EPP_J}[c][c][0.9]{\footnotesize \sf J-EPP \hspace{1em}}

\psfrag{Jeffreys}[c][c][0.9]{\footnotesize \textsf{Jeffreys} \hspace{1em}}

\includegraphics[ scale = 0.65 ]{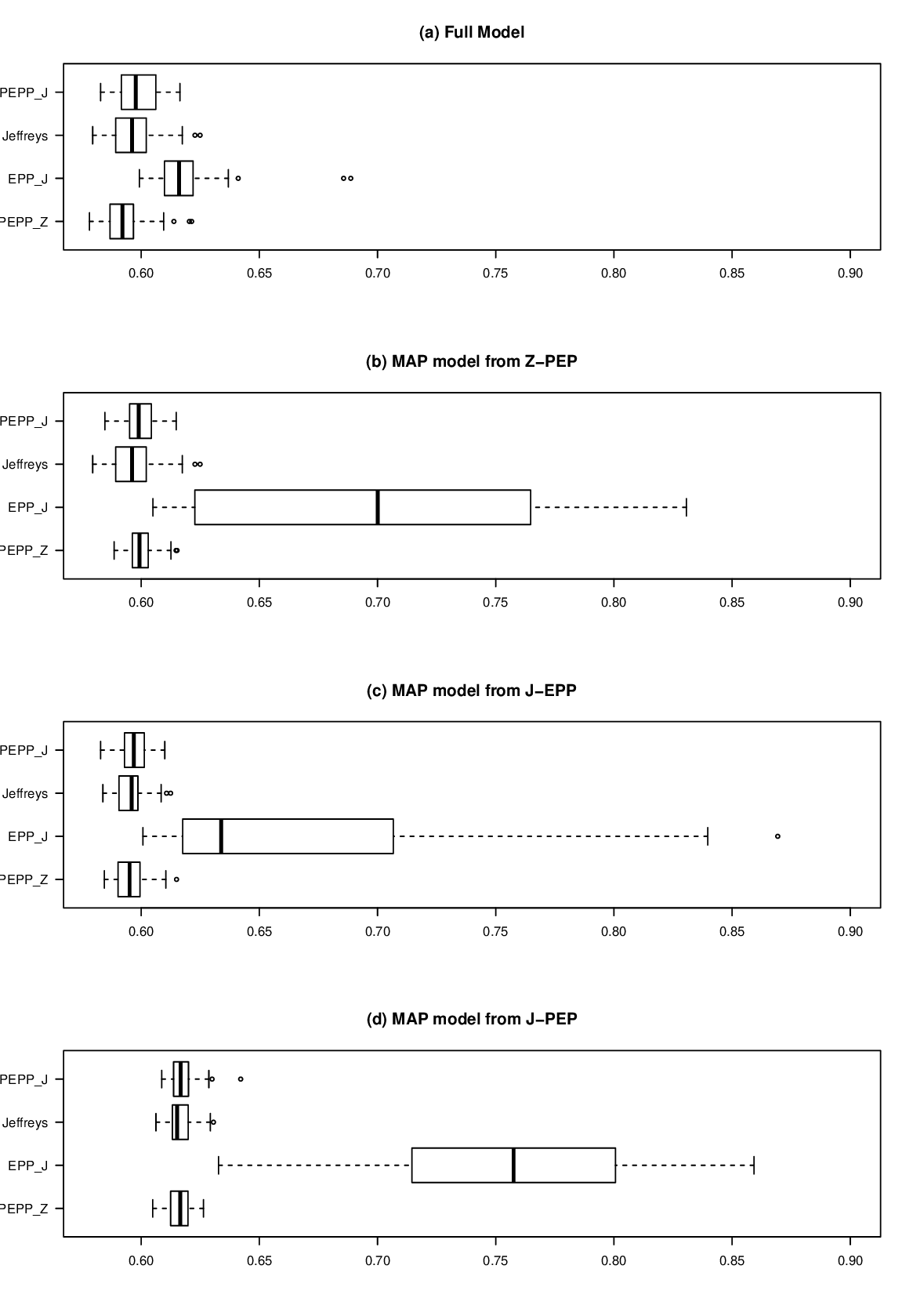}

\end{figure}

\subsubsection{Comparison of predictive performance}

\label{comparison-predictive-performance}

Here we examine the out-of-sample predictive performance of J-PEP, Z-PEP
and J-EPP on the full model and the three MAP models found by each method
implemented in the previous analysis. To do so, we randomly partitioned the
data in half 50 times, referring to the partition sets as modeling
($\mathbb{ M }$) and validation ($\mathbb{ V }$) subsamples. For each
partition, we generated an MCMC sample of $T = 1,000$ iterations from the
model of interest $M_\ell$ (fit to the modeling data $\mathbb{ M }$) and
then computed the following measure of predictive accuracy:
\begin{equation} \label{RMSE}
RMSE_\ell = \sqrt{ \frac{ 1 }{ T } \sum_{ t = 1 }^T \frac{ 1 }{ n_V }
\sum_{ i \in \mathbb{ V } } \big( y_i - \widehat{ y }_{ i | M_\ell }^{ \,
( t ) } \big)^2 } \, ,
\end{equation}
the root mean squared error for the validation data set $\mathbb{ V }$ of
size $n_V = \lceil \, \frac{ n }{ 2 } \, \rceil$; here $\widehat{ y }_{ i |
M_\ell }^{ \, ( t ) } = \mt{ X }_{ \ell ( i ) } \, \dn{ \beta }_\ell^{ ( t
) }$ is the predicted value of $y_i$ according to the assumed model $\ell$
for iteration $t$, $\dn{ \beta }_\ell^{ ( t ) }$ is the vector of model
$M_\ell$ parameters for iteration $t$ and $\mt{ X }_{ \ell ( i ) }$ is the
$i$th row of the matrix $\mt{ X }_\ell$ of model $M_\ell$.

Results for the full model and the MAP models are given in Table
\ref{ozonetab5}. For comparison purposes, we have also included the
split-half $RMSE$ measures for these three models using predictions based
on direct fitting of model (\ref{new2-1}) with the independence Jeffreys
prior $f ( \dn{ \beta }_\ell \, , \sigma_\ell^2 ) \propto \frac{ 1 }{
\sigma_\ell^2 }$, which can be viewed as a parametric bootstrap approach
around the MLE for $\dn{ \beta }_\ell$ and the unbiased estimate of
$\sigma_\ell^2$, allowing for variability based on their standard errors.

Table \ref{ozonetab5} shows that all $RMSE$ values for the PEP and
Jeffreys-prior approaches are similar, indicating that PEP provides
predictive performance equivalent to that offered by the Jeffreys prior;
also note that the PEP and the Jeffreys-prior $RMSE$s for the two PEP MAP
models are close to the corresponding values for the full model, which has
considerably higher dimension. (The point of this comparison is to
demonstrate that the PEP approach, which can be used for variable
selection, achieves a level of predictive accuracy comparable to that of
the Jeffreys-prior approach, which cannot be used for variable selection
because of its impropriety.)

In contrast, with the J-EPP approach the $RMSE$ values of all four models
are noticeably higher than the corresponding values for the Jeffreys-prior
and PEP approaches. Figure \ref{armse1} provides the explanation, by
showing the distribution of $RMSE$ values across the 50 random data splits,
for each of the four implementations in each of the four models examined in
Table \ref{ozonetab5}. The J-EPP approach is predictively unstable as a
function of its training samples, an undesirable behavior that PEP's
performance does not share.

To round out the full picture, we also examined the predictive ability of
median probability (MP) models. The MP models under both the Z-PEP and EIBF
approaches turned out to be the same as the corresponding MAP models. 
Under the J-PEP approach, the MP model was of a slightly higher dimension
than the corresponding J-PEP MAP model (it coincided with the Z-PEP MAP
model except for the addition of covariate $X_{ 20 }$). Thus, in this
empirical study, the predictive performance of MP models was similar to
that of the MAP models depicted in Figure \ref{armse1}.  

\subsection{A simulation comparison with other methods}

\label{sim_st}

We conclude our experimental results with a simulation comparison of Z-PEP
with a variety of other variable-selection and shrinkage methods, as
follows: the $g$-prior \citep{zellner_76}, the hyper-$g$ prior
\citep{liang_etal_2008}, non-local priors \citep{johnson_rossell_2010}, the
LASSO (least absolute shrinkage and selection operator; \citet{tibshirani})
and SCAD (smoothly-clipped absolute deviations; \citet{fan_li_2001}). (Note
that LASSO and SCAD are not focused on model selection but on the shrinkage
of coefficients; this feature can produce good point estimates and
prediction, but it precludes selection of a best subset (for a similar
argument see \cite{womack_2014}.)) For the $g$-prior and hyper-$g$ prior we
used the \texttt{BAS} package in \texttt{R}; we set $g = n$ in the former to
correspond to the unit information prior \citep{kass_wasserman_95}, and with
the hyper-$g$ prior we used $\alpha = 3$, as recommended by
\cite{liang_etal_2008}. For the implementation of SCAD and LASSO we used the
\texttt{R} packages \texttt{ncvreg} and \texttt{parcor}, respectively; in
both cases the shrinkage parameters were specified using
10-fold-cross-validation. Finally, following \cite{johnson_rossell_2012},
for the non-local prior densities we used the product moment (pMOM)
densities of the first and second orders ($r = 1$ and 2, respectively) and
the product inverse moment (piMOM) density, as implemented in the \texttt{R}
package \texttt{mombf}. All of these \texttt{R} routines are available at
\url{http://cran.r-project.org/web/packages}.  

We compared all eight methods on the Nott-Kohn case study (\ref{new3-1},
\ref{ss}) with the 50 additional data sets examined in Section
\ref{sensitivity-analysis-for-n-star}, by calculating the out-of-sample
predictive $RMSE$ (equation \ref{RMSE}), using for each sample an
additional simulated set of data of the same size ($n_V = 50$). The $RMSE$
was computed for each data set based on posterior estimates of the MAP
model for each variable-selection method. For Z-PEP, the $g$-prior and the
hyper-$g$ prior we used the posterior means; for the non-local priors we
employed the posterior modes; and for the shrinkage methods we used the
final estimates produced. Figure \ref{ex3_rmse} depicts the distribution of
$RMSE$ across the 50 samples for all methods under comparison, and Figure
\ref{ex3_diff_rmse} presents the distribution of pairwise differences
between the Z-PEP $RMSE$s and those of the other methods. It is evident
that Z-PEP exhibited somewhat better predictive performance in relation to
all the other approaches in this simulation study: the proportions of data
sets in which Z-PEP had smaller $RMSE$s were (56\%, 60\%, 62\%, 64\%, 66\%,
70\% and 76\%) in relation to (the hyper-$g$ prior, pMOM with $r = 1$,
SCAD, LASSO, the $g$-prior, piMOM and pMOM with $r = 2$), respectively.

\begin{figure}[p!]

\centering

\caption{\textit{Boxplots (over 50 simulated Nott-Kohn samples) comparing
the out-of-sample $RMSE$s for eight variable selection and shrinkage
methods.}}

\label{ex3_rmse}

\vspace*{-0.1in}

\includegraphics[ scale = 0.5, angle = -90 ]{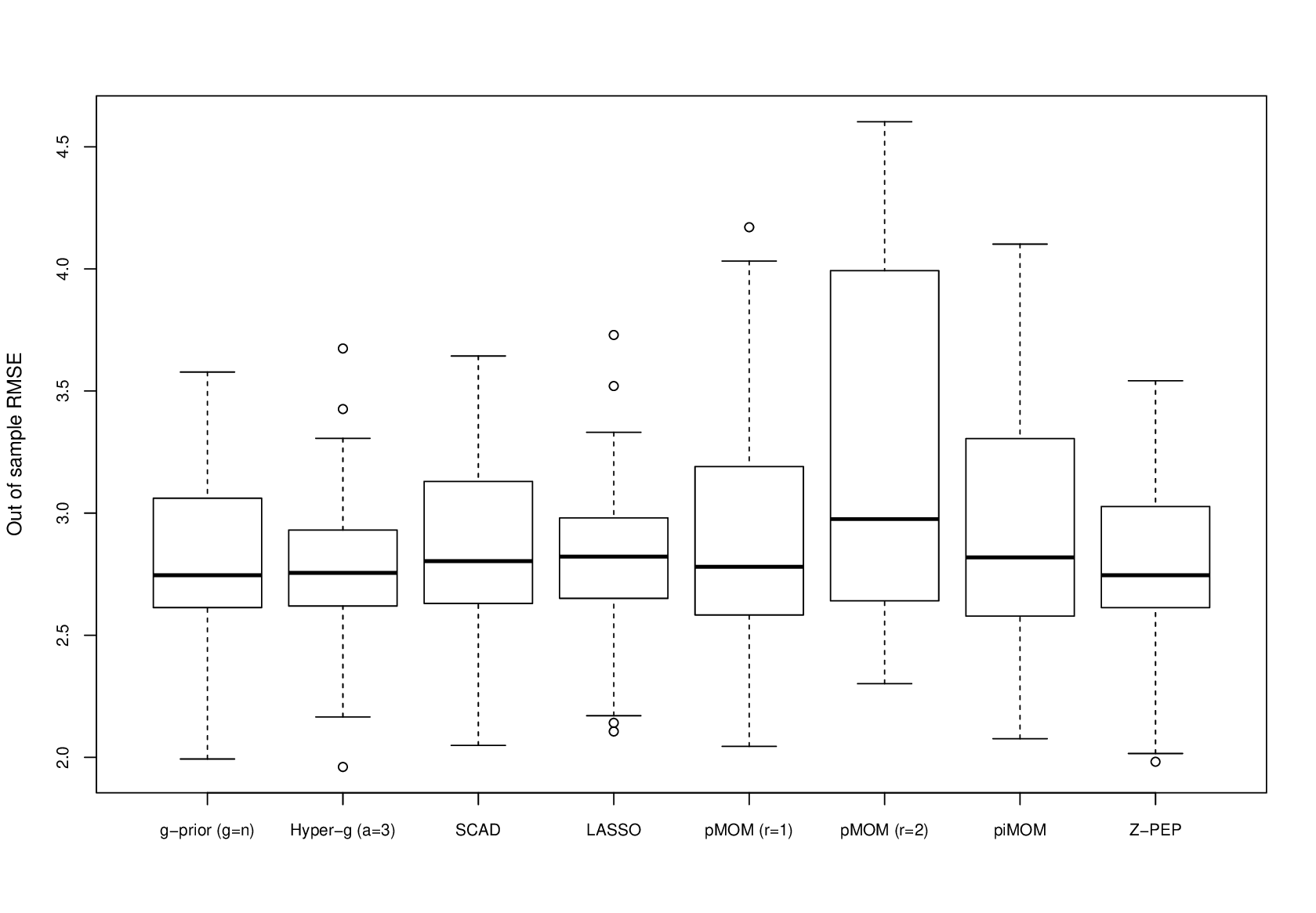}

\caption{\textit{Boxplots (over 50 simulated Nott-Kohn data sets) of the
differences between the out-of-sample $RMSE$s of seven variable selection
and shrinkage methods and the out-of-sample $RMSE$ of Z-PEP.}}

\label{ex3_diff_rmse}

\includegraphics[ scale = 0.5, angle = -90 ]{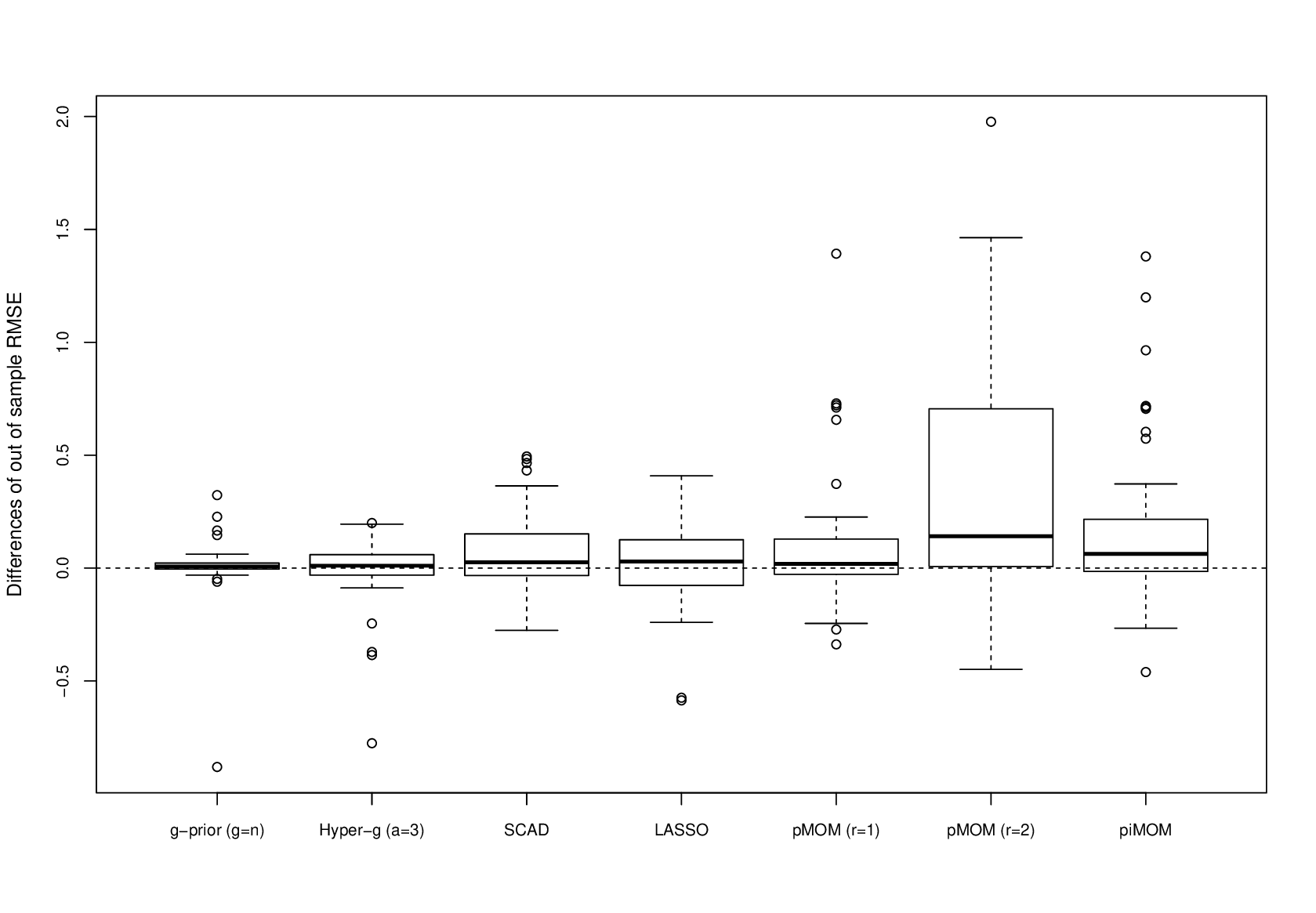}

\end{figure}

\begin{figure}[p!]

\caption{\textit{Proportions (across 50 simulated Nott-Kohn data sets) of
instances in which each covariate was identified with a non-zero effect by
the eight variable-selection and shrinkage methods under consideration.}}

\label{ex3_prop_inc_prob}

\centering

\vspace*{-0.325in}

\includegraphics[ scale = 0.55, angle = -90 ]{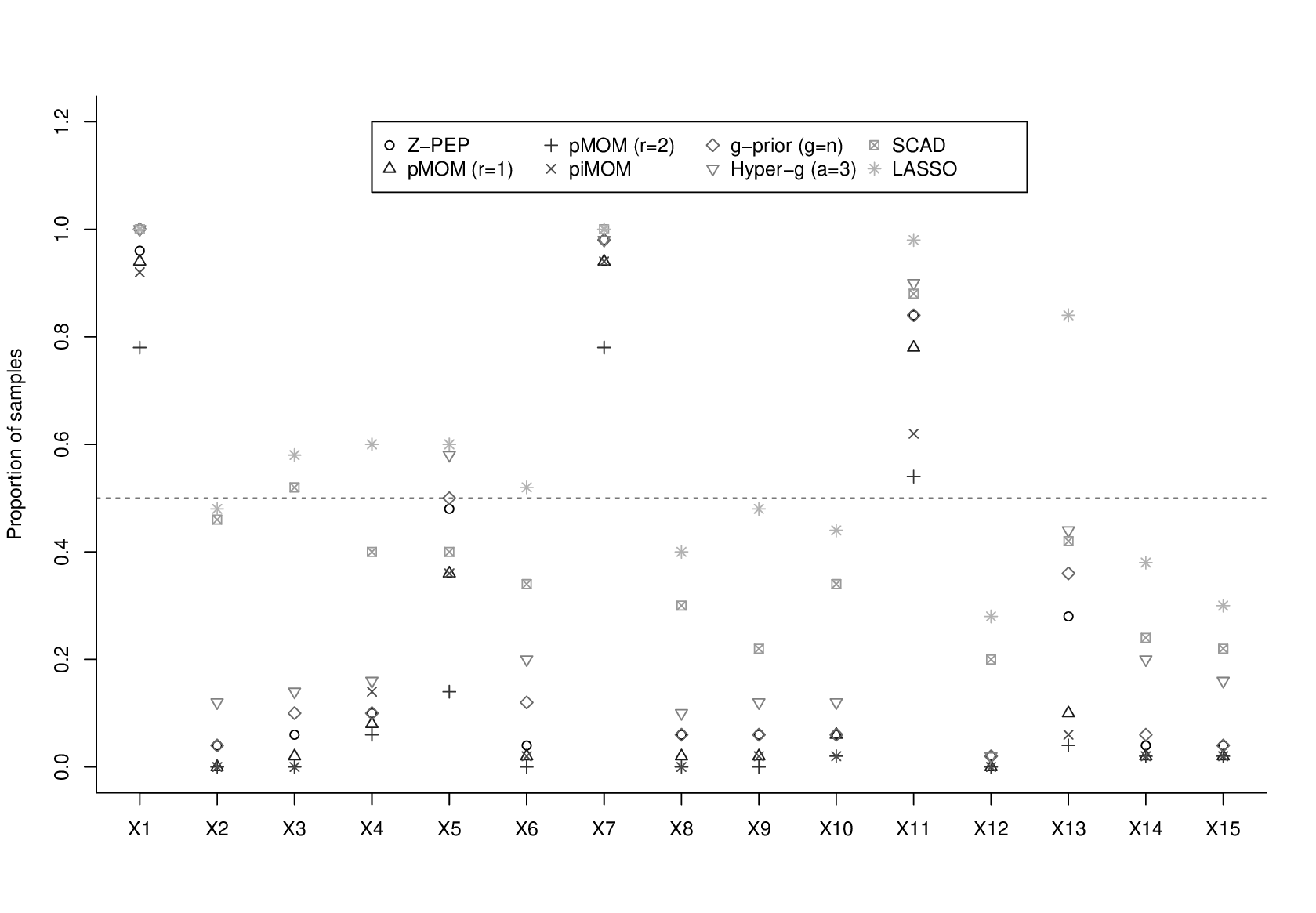}

\caption{\textit{Averages of posterior inclusion probabilities (across 50
simulated Nott-Kohn samples) for the six variable-selection methods under
consideration.}}

\label{ex3_mean_inc_prob}

\vspace*{-0.325in}

\includegraphics[ scale = 0.55, angle = -90 ]{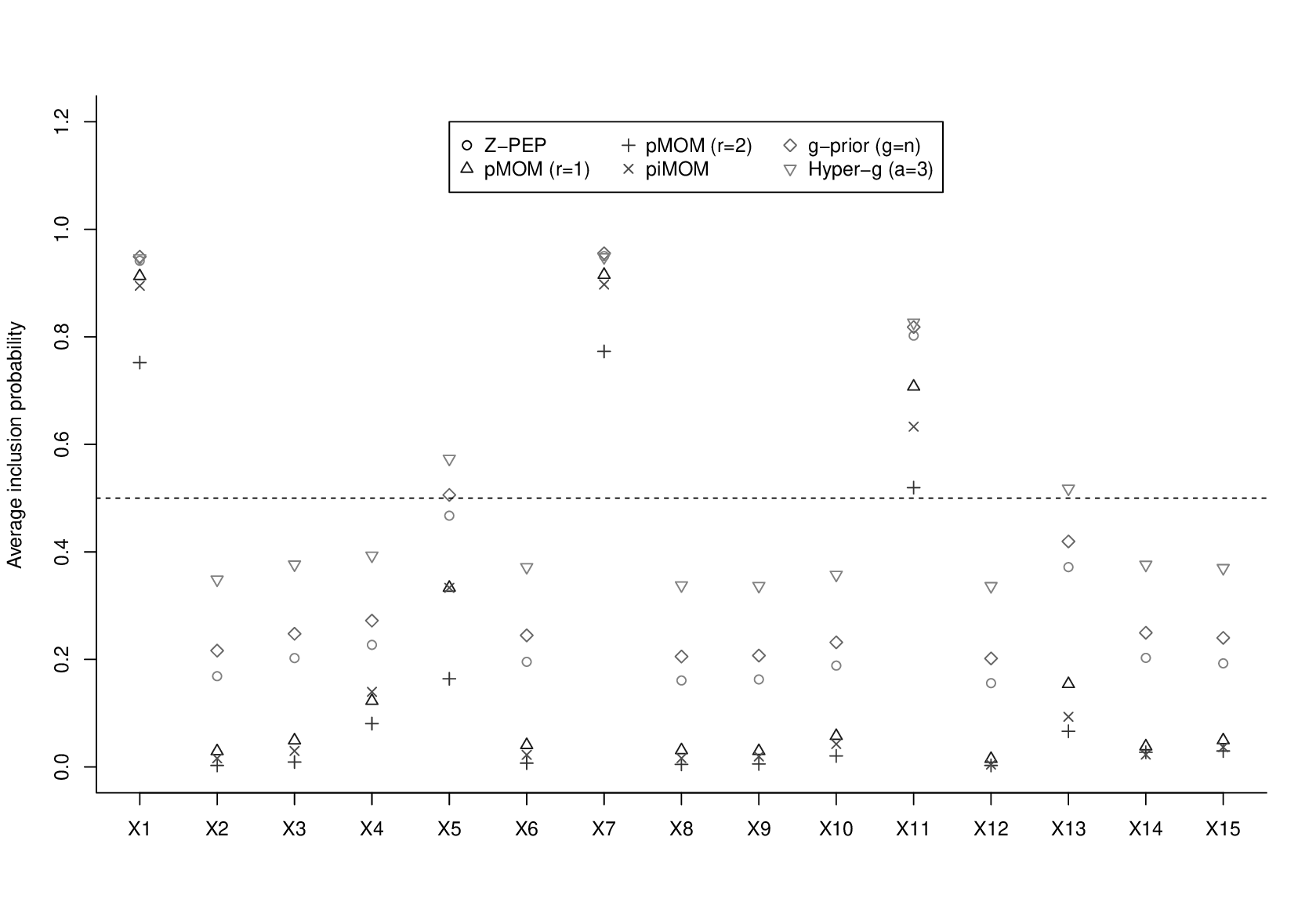}

\vspace*{-0.15in}

\end{figure}

We also examined all eight of the methods compared here with respect to
their variable-selection performance, in two ways: Figure
\ref{ex3_prop_inc_prob} presents the proportions (across the 50 simulated
Nott-Kohn data sets) of instances in which each covariate was identified
with a non-zero effect (i.e., the cases where (a) the effect was not
restricted to zero in the shrinkage methods and (b) the posterior inclusion
probabilities were found to be greater than 0.5 in the variable-selection
methods), and Figure \ref{ex3_mean_inc_prob} gives the mean posterior
variable-inclusion probabilities across the 50 replicate data sets. In the
following, we define the convention for each method that a variable is
\textit{selected} if either its proportion (Figure \ref{ex3_prop_inc_prob})
or mean inclusion probability (Figure \ref{ex3_mean_inc_prob}) exceeded
0.5.  Under this convention, all methods did well in finding the ``true"
covariates $X_1, X_7$ and $X_{ 11 }$, and in avoiding selection of the
``false" covariates $X_2, X_8, X_9, X_{ 10 }, X_{ 12 }, X_{ 14 }$ and $X_{
15 }$. Predictors $X_5$ and $X_{ 13 }$, which are built into the
data-generating model with smaller coefficients than those given to $X_1$
and $X_7$, were correctly selected only by the LASSO and the hyper-$g$
prior (in the case of $X_{ 5 }$ and $X_{ 13 }$) and the $g$-prior (in the
case of $X_5$); covariates $X_3, X_4$ and $X_6$, which have data-generating
coefficients of 0 in the Nott-Kohn setting, were falsely selected by the
LASSO and SCAD. Evidently the LASSO achieves its superior true-positive
behavior in this case study only at the expense of an undesirably high
false-positive rate. Z-PEP's selection rates were nearly 50\% for $X_{ 5 }$
and 30--40\% for $X_{ 13 }$, making it competitive with (though somewhat
inferior to) the hyper-$g$ prior and the $g$-prior on variable-selection
behavior in this example, but (as noted above) this is balanced by Z-PEP's
better predictive performance.

\section{Discussion}

\label{sec_discussion}

The major contribution of the research presented here is to simultaneously
produce a minimally-informative prior and sharply diminish the effect of
training samples on previously-studied expected-posterior-prior (EPP)
methodology, resulting in a prior for variable selection in Gaussian
regression models with very good variable-selection accuracy and excellent
out-of-sample predictive behavior. As noted in the introduction, one of the
main advantages of EPPs is that they achieve prior compatibility across
models; the proposed prior in this paper also has this property (in contrast
to other priors that have been developed in the Bayesian model selection
literature, such as mixtures of $g$-priors), and in addition our prior has a
unit-information structure and is robust to the size of the training sample.
By combining ideas from the power-prior approach of \cite{ibrahim_chen_2000}
and the unit-information prior of \cite{kass_wasserman_95}, we raise the
likelihood involved in EPP to a power proportional to the inverse of the
training sample size, resulting in prior information equivalent to one data
point. In this way, with our power-expected-posterior (PEP) methodology, the
effect of the training sample is minimal, regardless of its sample size, and
we can choose training samples with size $n^*$ equal to the sample size $n$
of the original data, thus eliminating the need for training samples
altogether.  This choice promotes stability of the resulting Bayes factors,
removes the arbitrariness arising from individual training-sample
selections, and avoids the computational burden of averaging over many
training samples. Additional advantages of our approach over methods that
depend on training samples include the following.

\begin{itemize}

\item

In variable-selection problems in linear models, the training data refer to
both $y$ and $X$. Under the base-model approach (see Section
\ref{expected-posterior-priors}), we can simulate training data $y^*$
directly from the prior predictive distribution of a reference model, but
we still need to consider a subsample $X^*$ of the original design matrix
$X$. The number of possible subsamples of $X$ can be enormous, inducing
large variability, since some of those subsamples can be highly influential
for the posterior analysis. By using our approach, and working with
training-sample sizes equal to the size of the full data set, we avoid the
selection of such subsamples by choosing $X^* = X$.

\item

The number $p$ of covariates in the full model is usually regarded as
specifying the minimal training sample. This selection makes inference
within the current data set coherent, but the size of the minimal training
sample will change if additional covariates are added, meaning that the EPP
distribution will depend incoherently on $p$. Moreover, if the data derive
from a highly structured situation (such as an analysis of covariance in a
factorial design), most choices of a small part of the data to act as a
training sample would be untypical. Finally, the effect of the minimal
training sample will be large in settings where the sample size $n$ is not
much larger than $p$. This type of data set is common in settings (in
disciplines such as bioinformatics and economics) in which (i) cases (rows
in the data matrix) are expensive to obtain (bioinformatics) or limited by
the number of available quarters of data (economics) but (ii) many
covariates are inexpensive and readily available once the process of
measuring the cases begins.

\end{itemize}

It is worth noting that our method, which is intended for settings in which
there is a fixed covariate space of $p < n$ predictor variables, works in a
totally different fashion than fractional Bayes factors. In the latter, the
likelihood is partitioned based on two data subsets; one is used for
building the prior within each model and the other is employed for model
evaluation and comparison. In contrast, with our approach, the original
likelihood is used only once, for simultaneous variable selection and
posterior inference. Moreover, the fraction of the likelihood (power
likelihood) --- used in the expected-posterior expression of our prior
distribution --- refers solely to the imaginary data coming from a prior
predictive distribution based on the reference model.

Our PEP approach can be implemented under any baseline prior choice; results
using the $g$-prior and the independence Jeffreys prior as baseline choices
are presented here. The conjugacy structure of the $g$-prior in Gaussian
linear models makes calculations simpler and faster, and also offers
flexibility in situations in which non-diffuse parametric prior information
is available. When, by contrast, strong information about the parameters of
the competing models external to the present data set is not available, the
independence Jeffreys baseline prior can be viewed as a natural choice, and
noticeable computational acceleration is provided by the fact that the
posterior with the Jeffreys baseline is a special case of the posterior with
the $g$-prior as baseline. In the Jeffreys case we have proven that the
resulting variable-selection procedure is consistent; we conjecture that the
same is true with the $g$-prior, but the proof has so far been elusive.

From our empirical results in two case studies --- one involving simulated
data, the other a real example based on the prediction of air pollution
levels from meteorological covariates --- we conclude that our method

\begin{itemize}

\item

is systematically more parsimonious (under either baseline prior choice)
than the EPP approach using the Jeffreys prior as a baseline prior and
minimal training samples, while sacrificing no desirable performance
characteristics to achieve this parsimony;

\item

is robust to the size of the training sample, thus supporting the use of
the entire data set as a ``training sample" --- thereby eliminating the
need for random sampling over different training sub-samples, which
promotes inferential stability and fast computation; 

\item

identifies maximum a-posteriori models that achieve better out-of-sample
predictive performance than that attained by a wide variety of
previously-studied variable-selection and coefficient-shrinkage methods,
including standard EPPs, the $g$-prior, the hyper-$g$ prior, non-local
priors, the LASSO and SCAD; and

\item

has low impact on the posterior distribution even when $n$ is not much
larger than $p$.

\end{itemize}

Our PEP approach could be applied to any prior distribution that is defined
via imaginary training samples. Additional future extensions of our method
include implementation in generalized linear models, where computation is
more demanding.

\section*{Supplementary material}

The Appendix is available in a web supplement at ***.

\section*{Acknowledgments}

We wish to thank the Editor-in-Chief, an Editor, an Associate Editor and two
referees for comments that greatly strengthened the paper. This research has
been co-financed in part by the European Union (European Social Fund-ESF)
and by Greek national funds through the Operational Program ``Education and
Lifelong Learning" of the National Strategic Reference Framework
(NSRF)-Research Funding Program: Aristeia II/PEP-BVS.

\section*{Abbreviations used in the paper}

BIC = Bayesian information criterion, EIBF = expected intrinsic Bayes
factor, EPP = expected-posterior prior, IBF = intrinsic Bayes factor, J-EPP
= EPP with Jeffreys baseline prior, J-PEP = PEP prior with Jeffreys-prior
baseline, LASSO = least absolute shrinkage and selection operator, PEP =
power-expected-posterior, Z-PEP = PEP prior with Zellner $g$-prior baseline,
SCAD = smoothly-clipped absolute deviations.

\bibliographystyle{agsm}

\bibliography{biblio3}

\end{document}